\documentclass[a4paper,aps,floatfix]{revtex4}

\usepackage{amsmath,amsthm,amsfonts,amssymb,times,bbm,graphicx}

\usepackage{algorithm}
\usepackage{algorithmicx}
\usepackage{algpseudocode}

\usepackage{hyperref}
\hypersetup{
  colorlinks=true,
  hypertexnames=false,
  linktocpage,
  colorlinks=true, 
  urlcolor=magenta!90!black,    
  linkcolor=blue!60!black, 
  citecolor=black!60 
}

\usepackage{relsize}
\newtheorem{theorem}{Theorem}
\newtheorem{proposition}[theorem]{Proposition}
\newtheorem{lemma}[theorem]{Lemma}

\usepackage{xcolor}
\usepackage{mathtools}

\newcommand{\ket}[1]{\mbox{$| #1 \rangle$}}

\usepackage{comment}
\DeclareMathOperator\arctanh{arctanh}
\usepackage{enumitem}

\usepackage{tikz}\usepackage{tikzscale}
\usetikzlibrary{quantikz}
\usepackage[many]{tcolorbox}    	

\usepackage{color}
\definecolor{myred}{rgb}{1,0,0}

\definecolor{myblue}{rgb}{0,0,0.8}

\definecolor{myyellow}{rgb}{0.9,0.8,0}

\definecolor{mygreen}{rgb}{0,0.7,0}

\definecolor{myorange}{rgb}{1,1,0}

\begin{document}

\title{Runtime--coherence trade-offs for hybrid SAT-solvers}

\author{
	Vahideh Eshaghian$^{(1)}$, 
S\"oren Wilkening$^{(2), (3)}$, Johan {\AA}berg$^{(1)}$, David Gross$^{(1)}$
}
\email{vahideh@thp.uni-koeln.de}
\affiliation{
	(1) Institute for Theoretical Physics, University of Cologne, Germany; 
	(2) Institut f\"ur Theoretische Physik, Leibniz Universit\"at Hannover, Germany;
	(3) Volkswagen AG, Berliner Ring 2, 38440 Wolfsburg;
}

\date{\today}

\begin{abstract}
	Many search-based quantum algorithms that achieve a theoretical speedup are not practically relevant since they require extraordinarily long coherence times, or lack the parallelizability of their classical counterparts. 
	This raises the question of how to divide computational tasks into  a collection of parallelizable sub-problems, each of which can be solved by a quantum computer with limited coherence time. 
	Here, we approach this question via hybrid algorithms for the $k$-SAT problem. 
	Our analysis is based on Sch\"oning's algorithm, which solves instances of $k$-SAT by performing random walks in the space of potential assignments. 
	The search space of the walk allows for ``natural'' partitions, where we subject only one part of the partition to a Grover search, while the rest is sampled classically, thus resulting in a hybrid scheme. 
	In this setting, we argue that there exists a simple trade-off relation between the total runtime and the coherence-time, which no such partition based hybrid-scheme can surpass. 
	For several concrete choices of partitions, we explicitly determine the specific runtime coherence-time relations, and show saturation of the ideal trade-off. 
	Finally, we present numerical simulations which suggest additional flexibility in implementing hybrid algorithms with optimal trade-off.  
\end{abstract}

\maketitle

\section{Introduction}

Consider a quantum algorithm that takes exponential time to run, but still offers a polynomial speedup over the best classical method.
Examples include Grover searches to brute-force a password or for finding the solution  for a combinatorial optimization problem for which no classical heuristics exist.
Fully quantum implementations might not be desirable for two reasons:
(1) Quantum hardware that can sustain very long computations might not be available,
and
(2) quantum algorithms, like Grover's search, might not be easily amenable to parallelization.
One is thus lead to the question of how to best break up such instances into a set of smaller, parallelizable subproblems that can individually be solved on quantum hardware.

We consider the well-known \emph{satisfiability problem} with $k$ the number of literals in each clause, ($k$-SAT) and focus particularly on $3$-SAT since it provides an attractive test bed to investigate such questions. $k$-SAT is the archetypical combinatorial optimization problem and represents a class of use cases with considerable practical relevance. 
Moreover, there is a classical randomized algorithm \cite{Schoening99,SchoeningToranBook}
due to Sch\"oning, with a performance close to the best-known algorithms with provable performance, and which furthermore allows for a closed-form asymptotic run-time analysis. And indeed, the algorithm obtained by replacing the classical search of the Sch\"oning-procedure by a Grover search \cite{Grover96} yields a \emph{quantum-Sch\"oning algorithm} with a quadratic improvement \emph{vis-\`a-vis} its classical counterpart \cite{Ambainis04}. 
(Below, we will refer to quantum algorithms that arise this way as \emph{Groverizations} of their classical versions).

However, such `fully quantized' Sch\"oning's SAT-solvers cannot be performed in parallel, which arguably is a relevant feature for algorithms that run in exponential time. Hybrid schemes, based on `partial' Groverizations of Sch\"oning's algorithm, where Grover search procedures are applied only to certain sub-routines, usually do allow for parallelizations.

Starting point of our analysis is the stochastic nature of Sch\"oning's algorithm as a random walk. 
This point of view yields two classes of hybrid algorithms, where one class Groverizes the random choice of the initial state of the walk, while the other class Groverizes the randomness in the walk itself. 
Within an established model of Sch\"oning's algorithm, we optimize the resulting run-times by balancing the resources allocated to the subroutines.

\subsection{Runtime--Coherence Time Trade-Offs}
\label{sec:run_vs_coherence}

Before specializing the Sch\"oning process, let us briefly outline the trade-offs between runtime and coherence time that can be expected for quantum search problems. 
Consider an algorithm that solves instances of size $n$ with runtime $T(n)$.
For exponential-time algorithms, we work with a somewhat coarser measure, the \emph{(asymptotic) runtime rate} 
\begin{align*}
	\gamma=\lim_{n\to\infty} \frac1n \log T(n),
\end{align*}
where we drop the base of the logarithm from here on; the base is $2$ unless explicitly stated otherwise.
In other words, $T\in O^*(2^{\gamma n})$, where 
$O^*$ denotes scaling behavior up to polynomial factors.
The aim is to trade it off against the \emph{coherence time} required to run the algorithm.
If $C(n)$ is the longest time over which coherence has to be maintained while running the algorithm, then the \emph{coherence time rate} is
\begin{align*}
	\chi=\lim_{n\to\infty} \frac1n \log C(n).
\end{align*}

Now restrict attention to search algorithms with classical runtime rate $\gamma_{\mathrm{C}}$.
A completely Groverized version runs with rate 
$\gamma_{\mathrm{G}}=\gamma_{\mathrm{C}}/2$.
All of its runtime will be spent coherently, specifically executing Grover iterations.
Therefore, 
$\chi_{\mathrm{G}}=\gamma_{\mathrm{C}}/2$
as well.
We can visualize these two points in a ``runtime rate vs coherence time rate''-chart, a mode of visualization that we will employ frequently (Fig.~\ref{fig:chart}).

To achieve a trade-off between total runtime and coherence time, we will consider algorithms that apply Grover's procedure only to a subset of the search space.
It is easy to see that any algorithm which results from such a procedure must have coordinates $(\chi,\gamma)$ that lie on or above the line segment
\begin{align*}
	L=\big\{ (\chi, \gamma_{\mathrm{C}} - \chi) \,|\, \chi\in[0,\gamma_{\mathrm{C}}/2] \big\}
\end{align*}
that connects the purely classical point $(0,\gamma_{\mathrm{C}})$ to the completely Groverized one 
$( \gamma_{\mathrm{C}}/2, \gamma_{\mathrm{C}}/2)$.

Indeed, take a partial Groverization that achieves parameters $(\chi,\gamma)$.
Then one can replace the Grover part by a classical search.
The resulting classical algorithm will have parameters $(0,\gamma+\chi)$, because the Grover search contributed $\chi$ to the runtime rate, but its classical simulation will contribute $2\chi$ instead.
 But if the initial parameters were below the line, i.e.\ if $\gamma < \gamma_{\mathrm{C}} - \chi$, then the resulting classical algorithm runtime rate is
$\gamma + \chi < \gamma_{\mathrm{C}}$, 
contradicting the assumption that $\gamma_{\mathrm{C}}$ describes the classical complexity of the search.

\begin{figure}[ht]
    \begin{center}
			\includegraphics[scale=.65]{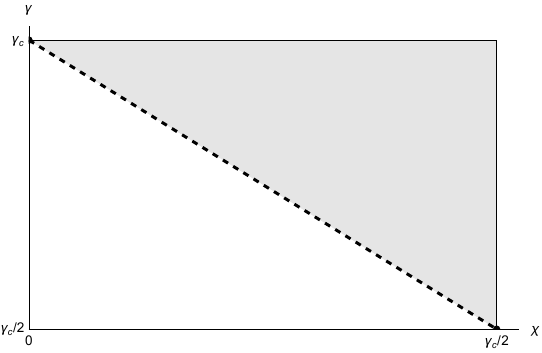}
    \end{center}
		\caption{\label{fig:chart} 
			We will frequently visualize the behavior of algorithms by indicating their position in a ``runtime rate vs coherence time rate''-chart.
			Classical algorithms require no coherence, and thus lie on the $y$-axis.
			In the example given, the point on the upper left hand side represents a classical probabilistic search with runtime rate $\gamma_{\mathrm{C}}$.
			A completely Groverized version has coordinates $( \gamma_{\mathrm{C}}/2, \gamma_{\mathrm{C}}/2)$ (bottom right), meaning that it will spend its entire runtime coherently.
			Hybrid algorithms that use Grover only for a subset of the search space must lie in the shaded area above or on the dashed line segment connecting these two points.
		}
\end{figure}

It is not obvious that, conversely, every point on this optimal line segment can actually be realized, much less with an algorithm that is ``natural'' or easy to implement.
Deciding the parameter ranges for natural partial Groverizations of Sch\"oning's procedure is the main goal of this paper.

\subsubsection{Related work}

Dunjko et al.\ \cite{Dunjko18} have previously considered partial Groverizations of Sch\"oning's algorithm.
They aimed to minimize a different metric:
total number of clean qubits, rather than coherence time.
In fact, they work in a highly constrained regime, where the number of available clean qubits only scales as $cn$, with $0< c <1$ and $n$ the number of variables of the given $3$-SAT formula. 
Surprisingly, they show that even this meager allotment of qubits in principle yields a speed-up compared to the classical Sch\"oning's algorithm \footnote{According to \cite{Dunjko18}, Supplemental Material Section~B.4, the relative speed-up to the classical Sch\"oning's rate is $f(c)= (1-\log\sqrt{3})\beta(c)$, where the Beta function up to $O(\frac{\log n }{n})$ is implicitly given as $A\beta(c)\ln{\frac{1}{\beta(c)}}+B\beta(c) = c$. As mentioned in \cite{Dunjko18} using a straightforward encoding of each trit into two qubits, one can assume $A= 10$ and $B=50$. To be consistent with our encoding, we consider $\log_23$ qubits to encode a trit and then, calculate the maximum speed-up in the rate, i.e. $f(1)\approx 0.0028$.}.

Despite the superficial similarities, their and our papers are quite different.
We allow for qubit-counts that are quasi-linear in $n$, i.e.\ $O (n\log n)$, reasoning that for exponential-time algorithms, coherence time and parallelizability might be more limiting than the number of available qubits.
As it will turn out, the setting considered here can interpolate between the classical and the fully Groverized performance, while the runtime rates obtainable in \cite{Dunjko18} stay close to the classical ones.
While \cite{Dunjko18} uses de-randomization techniques, our approach builds more directly on the original Sch\"oning's algorithm. This makes our approach technically less involved, and it also makes the lessons learned more widely applicable, since the basic technique of using Grover search over a subset of all variables, directly generalizes to any NP problem, whereas de-randomizations to a larger extent rely on the particular structure of the problem at hand.

\section{Setting the stage}
\subsection{Sch\"oning's algorithm}

Here we provide a very brief introduction to the pertinent aspects of Sch\"oning's $3$-SAT solver. 
For a more thorough review, we refer the reader to \cite{Schoening99,SchoeningToranBook}. 
In the $3$-SAT problem, we are given a collection of \emph{clauses} $C_1,\ldots, C_L$ on $n$ binary variables, where each clause is of the form $C_j = l_0^{(j)}\vee l_1^{(j)}\vee l_2^{(j)}$, and where each of the \emph{literals} $l_0^{(j)}$,  $l_1^{(j)}$, $l_2^{(j)}$ is one of the binary variables or its negation. 
The $3$-SAT formula is the conjunction of all the given clauses, $C:= \wedge_{j=1}^{L}C_j$, and the computational task is to determine whether there exists an assignment of the $n$ binary variables that satisfies $C$. 
According to Sch\"oning \cite{Schoening99}, an algorithm exists that, although with run-time that is exponential in $n$, can perform better than an exhaustive search through all potential assignments.  

Sch\"oning's algorithm (Alg.~\ref{alg:schoen}) depends on two parameters $N, m$ to be determined later.
It begins by choosing an assignment $x\in \{0,1\}^{\times n}$ uniformly at random.
The algorithm then performs an $m$-step \emph{random walk} over the space of $n$-bit strings (the inner loop in Alg.~\ref{alg:schoen}, from Line~\ref{line:walk}).
In every step, it checks (according to a pre-determined order) all the clauses $C_1,\ldots, C_L$. 
If all are satisfied, then $x$ is a solution and the algorithm terminates. 
Otherwise, it finds the first unsatisfied clause,
chooses one of the three variables corresponding to the literals of that clause uniformly at random.
The value of $x$ is then updated, by negating that variable. 
This concludes the step.
If no solution is found after $m$ steps, the walk is terminated.
Up to $N$ such walks are attempted (the outer loop in Alg.~\ref{alg:schoen}), each time using a fresh uniformly random starting point $x$.

\begin{algorithm}[H]
	\caption{Sch\"oning's Algorithm}\label{alg:schoen}
	\begin{algorithmic}[1]
		\Function{Schoening}{$C_1, \dots, C_L, N, m$}
			\For{$i=1 ...  N$}
				\State $x\leftarrow$ uniformly random value from $\{0,1\}^{\times n}$
				\For{$j=1...m$}
				\If{$x$ satisfies $C_1, \dots, C_L$}\label{line:walk}
						\State\Return $x$
					\Else
						\State $k\leftarrow$ index of first unsatisfied clause
						\State $l\leftarrow$ index of one of the three variables occurring in $C_k$, chosen uniformly at random
						\State $x\leftarrow x$, with the $l$-th bit of $x$ flipped
					\EndIf
				\EndFor
				\EndFor
			\State \Return False
	\EndFunction
	\end{algorithmic}
\end{algorithm}

\subsection{\label{SettingStageRunTime} Analysis of the run-time of Sch\"oning's algorithm}

The analysis of the run-time of Sch\"oning's algorithm is sketched in \cite{Schoening99,SchoeningToranBook} and a more in-depth analysis can be found in \cite{SwissPhDThesis}. 
Here we follow a very similar line of reasoning, with our particular ansatz in mind. 
In the following, we present an overview, see Appendix~\ref{AppFromSchToZ} for a more detailed account.

Assume that there is at least one satisfying assignment $x^\star$.
We first aim to lower-bound the probability that a given random walk finds a solution.
Let $x_0$ be the (random) initial configuration, and $x_l$ the one attained after the $l$-th step of the random walk.
The probability that \emph{any} solution is found during \emph{any} step of the walk is certainly at least as large as the probability $P(x_m = x^\star)$ that the walk finds $x^\star$ at the $m$-th step.
To analyze $P(x_m=x^\star)$, we follow in the steps of Sch\"oning \cite{Schoening99,SchoeningToranBook}, and focus on the evolution of the \emph{Hamming-distance} $d_H(x_l, x^\star)$ between the current configuration and the selected satisfying assignment $x^\star$. 

The fundamental insight is that if a clause $C_k$ is violated at the $l$-th step, then at least one of the three variables that appear in $C_k$ must differ between $x_l$ and the satisfying assignment $x^\star$.
Thus, the random flip decreases the Hamming distance to the solution with probability at least $1/3$:
\begin{align}\label{eqn:schoening_improvement}
	P\big( d_H(x_{l+1}, x^\star) = d_H(x_l, x^\star) - 1 \big)\geq \frac13.
\end{align}
This suggests to pass from a description of the process on bit strings to its projection
$x_l \mapsto d_H(x_l,x^\star)$ 
onto $\mathbb{N}$.
However, this would generally yield a process that would be no easier to analyze than the original one. 
One may for example note that although Sch\"oning-process  $(x_l)_l$ is Markovian on the space of bit-strings $\{0,1\}^{\times n}$, one cannot generally expect its projection $\big(d_H(x_l,x^\star)\big)_l$ to be Markovian on $\mathbb{N}$. 

The general idea for the analysis is to replace (via a coupling) the true projection $\big(d_H(x_l,x^\star)\big)_l$ with another process $(d_l)_l$ on $\mathbb{Z}$, which is Markovian and which moreover upper-bounds the true Hamming-distance,
\begin{equation}
	\label{FundamentalBound}
	d_H(x_l,x^\star)\leq d_l.
\end{equation}
More precisely, the Markov process $(d_l)_l$ is defined by the transition probabilities 
\begin{equation}
	\label{TransitionProbabilities}
	P(d_{l+1} = d_l +1) = \frac{2}{3},\quad P(d_{l+1} = d_l -1) = \frac{1}{3}.
\end{equation}

The transition probabilities (\ref{TransitionProbabilities}) can be interpreted as worst-case scenarios of each step in the Sch\"oning process.

From the bound (\ref{FundamentalBound}) it follows that $P(x_l = x^\star) \geq  P(d_l \leq 0)$. 
In other words, the success probability of the Sch\"oning-process is lower-bounded by the probability that the substitute-process $d_l$ reaches $0$.

Given the lower bound $P(d_m \leq 0)$ on the probability of success of each given walk, we expect at least one out of $N=1/P(d_m \leq 0)$ walks to find $x^\star$.
More precisely, if $\epsilon$ is the tolerated probability for failure, then the number of repetitions needed in order to find an existing solution satisfies
\begin{align}\label{eqn:sampling_overhead}
	N\geq 
	\frac{\log{\epsilon}}{\log(1-P(d_m\leq 0))}.
\end{align}

The required number $N$ of repetitions will be exponential in $n$.
It is then common to take a coarser point of view, and only analyze the corresponding \emph{rate} 

\begin{align}\label{eqn:rate}
	\gamma := -\lim_{n\rightarrow \infty}\frac{1}{n} \log 
	P(d_m\leq 0),
	\qquad\text{so that}\qquad
	N=O^{*}(2^{\gamma n}),
\end{align}
where $O^*$ denotes scaling behavior up to polynomial factors in order to achieve any constant probability of failure $\epsilon$.
With the choice $m = n$ (i.e., the termination time is equal to the number of variables) it turns out \cite{Schoening99,SchoeningToranBook} 
that 
$\gamma \leq \log\frac{4}{3}\approx 0.415$.

It is surprisingly technically difficult to rigorously derive the ``global bound'' $P(x_l=x^\star)\geq P(d_l \leq 0)$ 
from the  ``local bound'' (\ref{eqn:schoening_improvement}).
However, the Markovian version $(d_l)_l$ of the Hamming distance random walk is commonly accepted as a good (in fact, conservative) model of the Sch\"oning-process.
In the main body of this paper, we will therefore phrase our arguments in terms of that model.
More technical details on the relation between the two processes are given in Appendix~\ref{AppFromSchToZ}.

\subsection{Partial Groverizations: The general idea}

For random walks we naturally tend to think of the randomness as being generated whenever needed, like when we assign the initial state, or make the random choices along the path.
However, we can alternatively picture the walk as a deterministic process that is fed with an external random string $S$; a list from which it picks the next entry whenever a random choice is to be made. 
When the purpose of the walk is to find (an efficiently recognizable) solution to some computational problem, one can thus view the walk as a (deterministic) map that designates each input string $S$ as being ``successful'' or ``unsuccessful'', in the sense of the walk reaching the satisfying solution $x^\star$ or not. 
To this mapping, we can in principle apply a Grover-search procedure, since the walk (as well as the solution-recognition procedure) can be performed via reversible circuitry, and can thus also be implemented coherently. 

As described in the previous section, Sch\"oning's algorithm proceeds with an initialization, followed by a random walk on the space of $2^n$ assignments. 
The initialization requires $n$ bits of randomness, $S_I$, since the initial state is selected uniformly over all $2^n$ strings. 
A walk of length $m$ requires a string $S_W$ of $m\log3$ bits to encode the needed randomness. 
The $\log3$-factor is due to the fact that, at each step, the algorithm randomly selects which one of the three literals (of the first violated clause) should be flipped. 
An $m$-step Sch\"oning-walk can thus be viewed as a map from $S = (S_I,S_W)$ to a binary variable that tells us whether a satisfying assignment has been reached or not.

With a coherent circuit that implements this map, we can thus replace the uniformly distributed random variable $S$, with a uniform superposition over a corresponding number of qubits, and proceed via standard Grover-iterations \cite{Grover96}. 
We would expect such a procedure to yield a satisfying assignment at a run-time that scales as $O^{*}(2^{n{\gamma_G}})$ iterations, with $\gamma_G = \frac{1}{2}\log\frac{4}{3}\approx 0.208$ \cite{Ambainis04}, i.e., the standard quadratic speed-up. 
Up to a few constant qubits, one needs $n + (\log3 + \log L)m$ qubits to encode this map as a quantum circuit, where $L$ is the number of clauses in the $3$-SAT formula (more details are given in Section \ref{circuit}). Since the number of clauses grows linearly in $n$ for the regime of interest by the SAT phase-transition conjecture \cite{SATtransition}, and for the Sch\"oning walk $m = n$, the space complexity of such encoding is $O(n\log n)$.

The view of random walks as maps on random input strings opens up for the concept of partial Groverizations. 
Nothing would in principle prevent us from regarding only a \emph{part} of the input string $S$ as the input of the Grover-procedure, while keeping the rest of the string classical. 
Needless to say, one would generally expect the result to be less efficient than the ``full'' Groverization. 
However, the gain would be that the partial Groverization breaks the tasks into a collection of subproblems, each of which can be run in parallel on a quantum device that requires shorter coherence time. 

Although it seems reasonable to expect that such a division in principle is always possible, one may also expect that it in general would be challenging to find a quantum circuit that implements it in an economical manner. 
(We can always resort to a full coherent circuit for $S$ in its entirety, putting the ``classical part'' in a diagonal state.) 
However, there may be ``natural'' divisions of the process, which can be exploited. 
For  Sch\"oning's algorithm it is close to hand to consider the division $S = (S_I, S_W)$, i.e., the division of the required randomness into the initialization-part and the walk-part. 
One can thus consider two particularly natural classes of ``partial'' Groverizations of Sch\"oning's algorithm. 
For one of these, the \emph{Groverized Initialization (GI)}, the choice of the initial state is implemented coherently, while the walk is kept ``classical''. 
For the \emph{Groverized Walk (GW)}, the choice of initial state is kept classical, while the walk itself is performed coherently.

As described in Section \ref{SettingStageRunTime}, the actual analysis is based on the random walk $(d_l)_l$ on $\mathbb{Z}$, rather than the true Sch\"oning walk on strings in $\{0,1\}^{\times n}$. 
The idea is nevertheless the same; the required randomness is divided into the initialization and the walk \emph{per se}, resulting in GI- and GW-processes. 
As described in Section \ref{SettingStageRunTime}, the rate of the true Sch\"oning-process can be bounded by the rate of the substitute process $(d_l)_l$. It turns out that a similar argument can be made for GW (see Appendix \ref{AppFromSchToZ}), thus yielding a rigorous bound for the rate also in this case.  However, for the other processes we rather regard the $(d_l)_l$ process as a model of the genuine Sch\"oning-walk, without rigorous guarantees of analogous bounds.

\section{\label{partialgrover}Partial Groverizations}
The previous section introduced two types of partial Groverizations of Sch\"oning's algorithm, GI and GW, based on the division $S = (S_I, S_W)$, i.e. the initial and the walk randomness. In this section, we describe these schemes in detail and further discuss their ``fractional'' cases.

In the GI scheme, there is an outer loop that classically samples $S_W$, and is followed by a Grover-search inner loop over the space of all possible $S_I$. Similarly, GW starts with a classical outer loop that samples $S_I$ and is followed by a Grover-search inner loop over the space of all possible $S_W$ (this space is well-defined as the walk length is fixed).
We obtain Fractional Groverized Initialization (FGI) by adapting GI to a regime where only a fraction $z$ of the variables in the initialization can be searched coherently, with $0\le z \le1$. 
Fractional Groverized Walk (FGW) is similarly an adaption of GW to a regime where Grover-search can be performed on the randomness of walks of at most $m_q$ steps, with $0 \le m_q$. 
In both these fractional schemes, two classical outer loops contain a Grover-search inner loop.  
The algorithms introduced here depend on parameters ($N_1$, $N_2$, etc), that will be specified explicitly in Section \ref{sec:runtime}. 

All Grover searches will use an oracle derived from the function shown in Alg.~\ref{alg:oracle}:
It tests whether a Sch\"oning-walk with initial configuration $x\in\{0,1\}^n$ and walk randomness $w\in\{1,2,3\}^m$ will lead to a satisfying assignment.
For notational convenience, we let the elements of $w$ take ternary in values, with the interpretation that $w_l$ determines which of the three literals occurring in the first violated clause (if any) in step $l$ of the walk is flipped.
For a qubit-based implementation, it is not difficult to re-label the decision variables using $\lceil m \log 3\rceil$ binary variables.

\begin{algorithm}[H]
	\caption{Sch\"oning Walk \& Oracle}\label{alg:oracle}
	\begin{algorithmic}[1]
			\Function{Oracle}{$x_0,w$}
			\State \Return \textsc{True} if $\textsc{SchoeningWalk}(x_0,w)$ satisfies all clauses, else \textsc{False}
	\EndFunction
	\State
			\Function{SchoeningWalk}{$x,w$}  
				\For{$j=1...m$}
				\If{$x$ violates one of $C_1, \dots, C_L$}
						\State $k\leftarrow$ index of first unsatisfied clause
						\State $l\leftarrow$ index of the $w_j$-th variable occurring in $C_k$
						\State $x\leftarrow x$, with the $l$-th bit of $x$ flipped
					\EndIf
				\EndFor
				\State \Return $x$
	\EndFunction
	\end{algorithmic}
\end{algorithm}

For the different variants of partial Groverizations discussed below, we will fix a subset of arguments to the oracle, and consider it as a function of the remaining ones.
Fixed arguments will be denoted as subscripts, e.g.\ $\textsc{Oracle}_w: x\mapsto \textsc{Oracle}(x,w)$.
With these conventions, we have:

\begin{algorithm}[H]
	\caption{Groverized Initialization}\label{alg:gi}
	\begin{algorithmic}[1]
			\For{$i=1 ...  N_2$}
				\State $w\leftarrow$ uniformly random value from 
				$\{1, 2, 3\}^{\times m }$
				\State $x\leftarrow$ Grover-search for $\lfloor \sqrt{N_1}\rfloor$ iterations using \textsc{Oracle$_w$}()
				\If{$x$ satisfies all clauses}
					\State\Return $x$
				\EndIf
			\EndFor
	\end{algorithmic}
\end{algorithm}

\begin{algorithm}[H]
	\caption{Groverized Walk}\label{alg:gw}
	\begin{algorithmic}[1]
			\For{$i=1 ...  N_1$}
				\State $x_0\leftarrow$ uniformly random value from $\{0,1\}^{\times n}$
				\State $w\leftarrow$ Grover-search for $\lfloor \sqrt{N_2}\rfloor$ iterations using \textsc{Oracle$_{x_0}$}()
				\State $x\leftarrow \textsc{SchoeningWalk}(x_0,w)$ 
				\If{$x$ satisfies all clauses}
					\State\Return $x$
				\EndIf
			\EndFor
			\State \Return False
	\end{algorithmic}
\end{algorithm}

One may note that the Grover search in the Groverized walk only is guaranteed to succeed (with high probability) for a specific collection of initial states. The number of rounds $N_1$ of the outer loop is selected in such a way that it with high probability hits the set of advantageous initial states at least once, thus allowing the Grover-procedure to reach the satisfying assignment. Similar remarks apply to the other partial Groverizations.

Next, we discuss the ``fractional searches''. 
In the first one, the argument $x$ of the oracle is broken up as $x=(x_c, x_q)$ with $x_q$ taking $\lfloor z\cdot n \rfloor$ bits and $x_c$ being $\lceil (1-z)\cdot n \rceil$ bits long.
Here, $z\in[0,1]$ is a free parameter whose value will be determined below.

\begin{algorithm}[H]
	\caption{Fractional Groverized Initialization}\label{alg:fgi}
	\begin{algorithmic}[1]
			\For{$i=1 ...  N_2$}
				\State $w\leftarrow$ uniformly random value from $\{1, 2, 3\}^{\times m }$
				\For{$j=1 ...  N_1^{(c)}$}
				\State $x_c\leftarrow$ uniformly random value from $\{0,1\}^{\times \lceil(1-z)n\rceil}$
				\State $x_q\leftarrow$ Grover-search for $\Big\lfloor \sqrt{N_1^{(q)}}\Big\rfloor$ iterations using \textsc{Oracle$_{(x_c, w)}$}()
				\State $x=(x_c,x_q)$
					\If{$x$ satisfies all clauses} 
						\State\Return $x$
					\EndIf
				\EndFor
			\EndFor
			\State \Return False
	\end{algorithmic}
\end{algorithm}

The second fractional algorithm breaks up the walk randomness as $w=(w_c, w_q)$ with $w_c\in\{1,2,3\}^{m_c}$ and $w_q\in\{1,2,3\}^{m_q}$ respectively.
Again, the values of $m_c, m_q$ are chosen later.

\begin{algorithm}[H]
	\caption{Fractional Groverized Walk}\label{alg:fgw}
	\begin{algorithmic}[1]
			\For{$i=1 ...  N_1$}
				\State $x_0\leftarrow$ uniformly random value from $\{0,1\}^{\times n}$
				\For{$j=1 ...  N_2^{(c)}$}
					\State $w_c\leftarrow$ uniformly random value from $\{1, 2, 3\}^{\times m_c}$
					\State $w_q\leftarrow$ Grover-search for $\Big\lfloor \sqrt{N_2^{(q)}}\Big\rfloor$ iterations using \textsc{Oracle$_{(x_0,w_c)}$}()
					\State $w=(w_c,w_q)$
					\State $x\leftarrow \textsc{SchoeningWalk}(x_0,w)$ 
					\If{$x$ satisfies all clauses} 
						\State\Return $x$
					\EndIf
				\EndFor
			\EndFor
			\State \Return False
	\end{algorithmic}
\end{algorithm}

In the final algorithm, a fraction of $z\in[0,1]$ of both types of variables, the ones corresponding to the initialization and the ones corresponding to the walk, will be treated  quantum mechanically.

\begin{algorithm}[H]
	\caption{Evenly Fractionalized Grover}\label{alg:efg}
	\begin{algorithmic}[1]
		\For{$i=1 ...  N^{(c)}$}
		\State $x_c\leftarrow$ uniformly random value from $\{0,1\}^{\times \lceil (1-z) n\rceil}$
		\State $w_c\leftarrow$ uniformly random value from $\{1, 2, 3\}^{\times \lceil(1-z)m\rceil}$
			\State $(x_q,w_q)\leftarrow$ Grover-search for $\big\lfloor \sqrt{N^{(q)}}\big\rfloor$ iterations using \textsc{Oracle$_{(x_c,w_c)}$}()
			\State $w=(w_c,w_q)$
			\State $x_0=(x_c,x_q)$
			\State $x\leftarrow \textsc{SchoeningWalk}(x_0,w)$ 
			\If{$x$ satisfies all clauses} 
				\State\Return $x$
			\EndIf
		\EndFor
		\State \Return False
	\end{algorithmic}
\end{algorithm}

\section{Run-time Analysis}
\label{sec:runtime}

We will now lower-bound the probability of success of the various approaches.
As a preparation, in Sec.~\ref{sec:analysis_classical}, we give a brief account of the analysis of the classical case, before moving on to the Groverized versions in Sec.~\ref{sec:analysis_grover}.

\subsection{The classical Sch\"oning process}
\label{sec:analysis_classical}

The main ideas of the classical analysis are close to their presentation in Refs.~\cite{Schoening99,SchoeningToranBook}.
We work in the Markovian model $(d_l)_l$ for the behavior of the Hamming distances, as laid out in Sec.~\ref{SettingStageRunTime}.
Frequently, it will be convenient to measure quantities ``in units of $n$ or $m$''. For example, we will soon choose a number $\kappa\in[0,1]$ and assume that the initial value $d_0$ is equal to $\kappa n$. 
Of course, this only makes sense if $\kappa n$ is an integer.
In order to keep the notation clean, we will implicitly assume that such expressions have been rounded to the next integer.

Choose numbers $\kappa, \nu \in[0,1]$.
A given walk $(d_l)_l$ is certainly successful (in the sense that $d_m\leq 0$) if 
\begin{enumerate}
	\item
		The initial value is $d_0=\kappa n$,
	\item
    the random walk decreases the Hamming distance in exactly $\nu m$ of its $m$ steps, and
	\item
		the condition
		\begin{align}\label{eqn:walk_condition}
			\kappa n \leq (2\nu -1)m 
		\end{align}
		holds.
\end{enumerate}
Indeed, the right hand side of (\ref{eqn:walk_condition}) is the difference between the number of steps where the Hamming distance has been decreased, $\nu m$, and the number of steps where the Hamming distance has been increased, $(1-\nu )m$. 

For any fixed pair of values $\kappa,\nu$ subject to (\ref{eqn:walk_condition}), we will now compute the probability of this particular route to success.
Denote the first event by $E_1$ and the second event by $E_2$.
They occur with respective probabilities
\begin{align}\label{eqn:events_probs}
    P(E_1) &= \frac{1}{2^n}\binom{n}{\kappa n},
					 &
    P(E_2) &= \binom{m}{\nu m}\Big(\frac13\Big)^{\nu m}\Big(\frac23\Big)^{(1-\nu)m}.
\end{align}
Since the two events are independent, 
the success probability of the walk 
is lower-bounded by 

\begin{align}
    P(x_m=x^\star|\kappa)\ge
    P(d_m\le 0|x) &\ge P(E_1\land E_2) 
    = P(E_1)P(E_2) 
    = \frac{1}{2^n}\binom{n}{\kappa n} \binom{m}{\nu m}\Big(\frac13\Big)^{\nu m}\Big(\frac23\Big)^{(1-\nu)m}.
    \label{jointprob}
\end{align}
The various binomial coefficients can be conveniently related to entropies.
To this end, recall the definition
of the \emph{binary entropy function} 
\begin{equation*}
		H(p) = -p\log p - (1-p) \log(1-p)  \quad \text{for} \quad p \in [0,1],
\end{equation*}
and the \emph{relative entropy}
\begin{equation*}
	 D(p\parallel q) = - p\log q - (1-p)\log(1-q) - H(p)  \quad \text{for} \quad p,q \in [0,1]. 
\end{equation*}
Then using the well-known estimate \cite[Chapter 11.1]{cover}

\begin{equation*}
    \frac{1}{n+1}2^{nH(\kappa)} \le \binom{n}{\kappa n} \le 2^{nH(\kappa)},
\end{equation*}
Equation (\ref{jointprob}) can, after some straight-forward calculations, be concisely rewritten as
\begin{align}\label{eqn:entropic_bound}
    P(d_m\le 0|x) \gtrsim 2^{-(1-H(\kappa))n}2^{-D(\nu \parallel \frac13)m},
\end{align}
where $\gtrsim$ denotes an inequality holds asymptotically, up to a polynomial factor. 
Equation~(\ref{eqn:entropic_bound}) directly gives an upper bound on the rate $\gamma$ defined in (\ref{eqn:rate}). 
Since the rate expresses the logarithm of the complexity ``in units of $n$'', it makes sense to also express the length of the walk in terms of $\mu := m/n$.
Then:

\begin{align}
		\gamma
		= - \lim_{n \rightarrow \infty}\frac1n \log P(d_m\le 0|x)
		\leq
		1- H(\kappa) + \mu D(\nu\parallel 1/3)
		=:
    \gamma(\mu,\kappa,\nu).  \label{gamma}
\end{align}
In particular, the infimum of $\gamma(\mu,\kappa,\nu)$ subject to the constraints (\ref{eqn:walk_condition}) and $0 \leq \mu, 0\leq \nu,\kappa \leq 1$ is a valid bound for $\gamma$.
We will perform such optimizations explicitly for the partially Groverized versions in Sec.~\ref{sec:analysis_grover}.
For the classical procedure, we just state the final result:
\begin{align}\label{eqn:schoening_optimal}
	\mu = 1, \qquad \kappa=\frac13, \qquad \nu=\frac23, \qquad \gamma_{\mathrm{C}} = \log\frac43\simeq 0.4150.
\end{align}

Remark:
One might be tempted to search a tighter bound by summing the contributions to the probability of success that arise from all consistent values for $\mu,\kappa,\nu$, instead of just considering the extremal value.
However, the rate of a sum of exponentially processes is asymptotically determined by the rate of the dominating summand alone, i.e.\ for all collections of $\gamma_i>0$, it holds that
\begin{align*}
	 \lim_{n\to\infty}-\frac1n \log \sum_i 2^{-\gamma_i n} = \sup_{i} \gamma_i
\end{align*}
(assuming convergence).
Therefore, considering only the dominating term does not affect the overall asymptotic rate.

\subsection{Partially Groverized processes}
\label{sec:analysis_grover}

In this section, we derive the main results of this paper:
Bounds on the asymptotic rates for partially Groverized versions of Sch\"oning's scheme.

\subsubsection{Groverized Initialization, Algorithm~\ref{alg:gi}}
\label{sec:gi}

For the parameters $N_1, N_2$, we choose constant multiples of $1/P(E_1), 1/P(E_2)$ respectively.
The value of the constant depends on the acceptable probability $\epsilon$ of failure, as exhibited in Eq.~(\ref{eqn:sampling_overhead}).
Since this constant does not effect the rate, we will not specify it here.
The probabilities do depend essentially on the parameters $\mu,\kappa,\nu$, though. 
We will therefore write $N_1(\kappa)$ and $N_2(\mu,\nu)$.
Because the asymptotic complexity of a Grover search is the square root of the classical complexity,
the rate function of GI 
is then given by
\begin{align}\label{rateGI}
	\gamma_{\textrm{GI}}(\mu,\kappa,\nu) 
	&=  
	\lim_{n \rightarrow \infty}\frac1n \log \left(\sqrt{N_1(\kappa)}N_2(\nu,\mu)\right) 
	= 
	\frac{1 - H(\kappa)}2 
	+  
	\mu D(\nu \parallel 1/3). 
\end{align}
Likewise, the required coherence time scales with the number of Grover iterations, i.e.\ as $O^*(2^{\chi n})$, for
\begin{align}
    \chi(\kappa) 
		:=  \lim_{n \rightarrow \infty}\frac1n \log \sqrt{N_1(\kappa)}
		= \frac{1 - H(\kappa)}2.
\end{align}
The parameters are constrained by
\begin{equation}
	\label{constsGIGW}
	\begin{aligned}
			0 \le \kappa \le 1,
			\quad
			0 \le \mu,
			\quad
			0 \le \nu \le 1,
			\quad
			\frac{\kappa}{2\nu -1} \leq \mu,
	\end{aligned}
\end{equation}
where the final condition is a re-arranged version of the success criterion (\ref{eqn:walk_condition}).

We now determine the minimal rate $\gamma_{\mathrm{GI}}$ over the consistent parameters.
Because relative entropy is non-negative, it is always advantageous to reduce the value of $\mu$ until it is minimal subject to the constraints.
This is achieved by changing the final inequality in (\ref{constsGIGW}) to equality.
Re-arranging, we arrive at
\begin{equation}
	\label{constsGIGW2}
	\begin{aligned}
			0 \le \kappa \le 1,
			\quad
			0 \leq \mu,
			\quad
           \nu = \frac12+\frac{\kappa}{2\mu},
	\end{aligned}
\end{equation}
which allows us to eliminate $\nu=\nu(\kappa,\mu)$ from the problem.
Varying $\gamma$ with respect to $\mu$ gives rise to the criticality condition
\begin{align}
	0
	\stackrel{!}{=}
	\partial_\mu \gamma_{\textrm{GI}}(\kappa,\mu)
	= 
	\partial_\mu \,\mu D(1/2+\kappa/(2\mu) \parallel 1/3)
	=
	\frac12\log\Big(\frac{\mu+\kappa}{\mu}\Big) +\frac12\log\Big(\frac{\mu-\kappa}{\mu}\Big)+\log3-\frac32.
		\label{partialv} 
\end{align}
This can be solved explicitly e.g.\ using a computer algebra system \cite{our-data}, leading to
\begin{align}\label{eqn:mu_nu}
	\mu = 3\kappa
	\qquad\Rightarrow\qquad
	\nu=\frac23,
	\quad
	\mu D(\nu \parallel 1/3) = \kappa.
\end{align}

Eliminating $\mu$, we get
\begin{align}\label{eqn:gi_parametric}
	\begin{split} 
		\gamma_{\textrm{GI}}(\kappa)
		&=
		\frac{1-H(\kappa)}2 + \kappa, \\
		\chi_{\mathrm{GI}}(\kappa)
		&=
		\frac{1-H(\kappa)}2.
	\end{split}
\end{align}

The pair of equations (\ref{eqn:gi_parametric}) contain all information about the asymptotic behavior of the Groverized Initialization procedure.
Each value of $\kappa$ gives a solution for the two undetermined constants $N_1(\kappa), N_2(\mu=3\kappa, \nu=\frac23)$ in Alg.~\ref{alg:gi}, in such a way that it will run with a small probability of returning a false negative.
Varying $\kappa$, we thus obtain a family of algorithms that find different compromises between the required coherence time and the total runtime. 
The achievable pairs of values are shown in Fig.~\ref{fig:gi_tradeoffs}.

\begin{figure}[ht]
    \begin{center}
			\includegraphics[scale=.75]{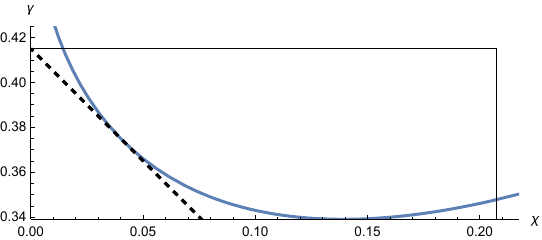}
    \end{center}
		\caption{\label{fig:gi_tradeoffs} 
			Rates $\big(\chi_{\mathrm{GI}}(\kappa), \gamma_{\mathrm{GI}}(\kappa)\big)$
			for the required coherence time and the total runtime of the Groverized Initialization algorithm, as the parameter $\kappa$ is varied.
			The horizontal bar denotes the runtime rate achieved by the classical Sch\"oning process. 
			In other words, points above this line are uninteresting.
			The vertical bar denotes the coherence rate that allows one to run a completely Groverized version of the Sch\"oning process.
			This, arguably, makes points to the right of this line uninteresting as well.
			Points to the left of the minimum (at $(\gamma,\chi)\simeq (0.339, 0.139)$) can represent advantageous choices if either the total coherence time of a quantum computer is limited, or a larger degree of parallelization is desired.
			The dashed line is the lower bound on the runtime rate given the coherence time, as introduced in Fig.~\ref{fig:chart}.
			It is achieved for $\kappa=\frac13$.
		}
\end{figure}

Finally, we explicitly determine the minimal rate achievable in the Groverized Initialization scheme.
With the help of a computer algebra system \cite{our-data}, one easily finds
\begin{align}\label{partialk}
    0 \stackrel{!}{=} 
		\partial_\kappa \gamma_{\textrm{GI}}(\kappa)=\frac12\log\frac{\kappa}{1-\kappa}+1
		\quad\Leftrightarrow\quad
		\log\Big(\frac1{\kappa}-1\Big) = 2 
		\quad\Rightarrow\quad
		\kappa = \frac15 
\end{align}
which gives
\begin{align}
	\label{eqn:gi_optimal}
	\mu = \frac35, 
	\qquad
	\gamma_{\textrm{GI}} = \frac{3-\log5}{2}\approx 0.339,
	\qquad
	\chi_{\textrm{GI}} \simeq 0.139.
\end{align}

Remark:
One can cast the final minimization into the form
\begin{align*}
	\gamma_{\textrm{GI}} = 
	\inf_{\kappa} \gamma_{\textrm{GI}}(\kappa)
	=
	\inf_{\kappa} 
	\left(
		\frac{1-H(\kappa)}2 + \kappa
	\right)
	=
	-
	\sup_{\kappa} 
	\left(
		-\kappa
		-	
		\frac{H(\kappa)-1}2
	\right).
\end{align*}
This expression shows that the optimization amounts to computing a Legendre transform.
Indeed, with $f(\kappa):=1/2(H(\kappa)-1)$, the right hand side equals $-f^*(-1)$.
For physicist readers, it might be amusing to note that $S(n \kappa)=n H(n \kappa)$ formally equals the entropy of an $n$-spin paramagnet as a function of the total magnetization.
The Legendre transform of the entropy is a Massieu thermodynamic potential, equal to $F/T$ (with $F$ the free energy) expressed as a function of the inverse temperature \cite[Chapter~5.4]{callen}.
We will, however, not pursue this analogy here.

\subsubsection{Groverized Walk, Algorithm~\ref{alg:gw}}

The analysis proceeds in close analogy to the above case.
The asymptotic rate function of GW is 
\begin{align}
    \gamma_{\textrm{GW}}(\kappa,\mu,\nu) 
		&= 
		\lim_{n \rightarrow \infty}\frac1n \log (N_1(\kappa)\sqrt{N_2(\nu,\mu)})  
    = 
		{1 - H(\kappa)} + \frac \mu{2} D(\nu \parallel 1/3) \label{rateGW},
\end{align}
subject to the set of constraints (\ref{constsGIGW}). 
The parameters $\nu, \mu$ can be treated in exactly the same way as before, leading again to (\ref{eqn:mu_nu}).
In particular, the coherence time rate takes the simple form $\chi=\kappa/2$, which allows us to eliminate $\kappa$ in favor of $\chi$.
We immediately obtain
\begin{align}
    \gamma_{\textrm{GW}}(\chi) 
		= 1-H(2\chi) + \chi.
\end{align}

Again, it is not difficult to solve for the lowest runtime \cite{our-data}: 
\begin{align}\label{eqn:gw_optimal}
	\mu
	=
  3(\sqrt{2}-1), 
	\qquad
	\kappa = \sqrt{2}-1,
	\qquad
	\gamma_{\textrm{GW}} \approx 0.228,
	\qquad
	\chi_{\textrm{GW}} 
	\simeq  0.2071.
\end{align}
At the optimal point, the runtime scales with a rate that is very close to the one of a full Groverization of Sch\"oning's process, namely $\gamma_{\mathrm{FG}}=\gamma_{\mathrm{C}}/2\simeq .2075$.
The flip side is that the required coherence times are basically identical:
\begin{align*}
	\chi_{\textrm{FG}} 
	-
	\chi_{\textrm{GW}} 
	\simeq
	0.0004.
\end{align*}

The findings are summarized in Fig.~\ref{fig:gw_tradeoffs}.

\begin{figure}[ht]
    \begin{center}
			\includegraphics[scale=.6]{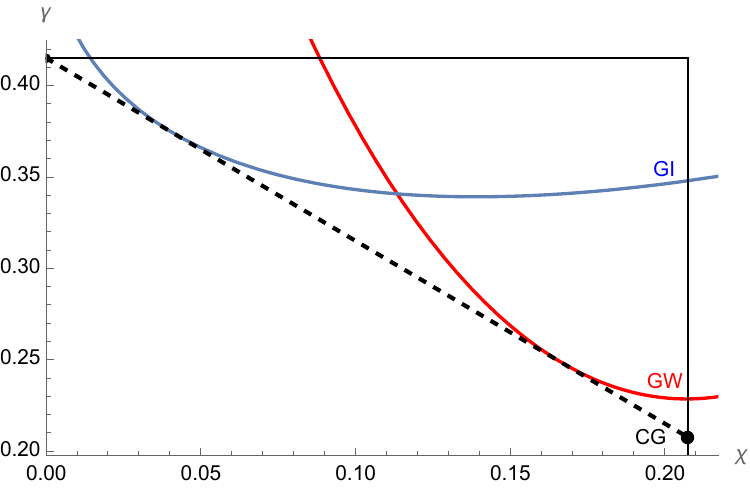}
    \end{center}
		\caption{\label{fig:gw_tradeoffs} 
			The runtime rate vs coherence time rate curves for 
			Groverized Initialization (GI, blue) and
			Groverized Walk (GW, red).
			The point marked ``CG'' at the bottom right of the diagram represents the complete Groverization of the Sch\"oning process.
			For long coherence times, GW is preferable, while for shorter coherence times GI achieves a lower total runtime.
		}
\end{figure}

\subsubsection{Fractional Groverized Initialization, Algorithm~\ref{alg:fgi}}
\label{sec:fgi}

In the case of Alg.~\ref{alg:fgi}, the initial Hamming distance is the sum of two terms $d_0 =  \kappa_c (1-z)n + \kappa_q z n$, which model $d_H(x_c,x_c^\star)$ and $d_H(x_q,x_q^\star)$ respectively.
Define the analogues 
\begin{align*}
	P(E_1^c) = \frac{1}{2^{(1-z)n}}\binom{(1-z)n}{\kappa_c (1-z)n},
	\qquad
	P(E_1^q) = \frac{1}{2^{z n}}\binom{z n}{\kappa_q  zn}
\end{align*}
of $P(E_1)$ introduced in Eq.~(\ref{eqn:events_probs}).
Analogous to the discussion in Sec.~\ref{sec:gi}, the parameters $N_1^c, N_1^q$ are defined as the reciprocals of these probabilities, times a constant that influences the probability of a false negative, but will not be discussed as it has no impact on the asymptotic rates.
The success criterion is now
\begin{align*}
	(1-z)\kappa_c + z\kappa_q &\leq (2\nu-1)\mu
\end{align*}
and the other constraints are
\begin{align*}
		0 \le \kappa_c, \kappa_q, \nu, z \le 1, 
		\qquad
		0 \le \mu.
\end{align*}
The asymptotic rate function for the runtime of FGI reads
\begin{align}
    \gamma_{\textrm{FGI}}(\kappa_c,\kappa_q,\nu,\mu;z) 
		&= \lim_{n\rightarrow \infty} \frac1n\log\Big(N_1^c(\kappa_c;z)\sqrt{N_1^q(\kappa_q;z)}N_2(\nu,\mu)\Big) \nonumber\\
    &= (1-z)\left(1- H(\kappa_c)\right) + \frac z2\big(1- H(\kappa_q)\big) + \mu D(\nu \parallel 1/3)
\end{align}

Arguing as in Sec.~\ref{sec:gi}, the inequality in the success criterion may be replaced by an equality.
Solving for $\nu$ gives
\begin{align*}
	\nu= \frac12 + \frac{(1-z)\kappa_c + z\kappa_q}{2\mu}.
\end{align*}
We proceed as in the first two cases. 
Criticality of $\partial_{\mu}\gamma_{\textrm{FGI}}$ with respect to $\mu$ occurs at

\begin{align*}
	\mu = 3((1-z)\kappa_c + z \kappa_q)
	\qquad\Rightarrow\qquad
	\mu D(\nu \parallel 1/3) = (1-z)\kappa_c + z \kappa_q, \quad \nu=\frac23.
\end{align*}
Plugging in, we arrive at
\begin{align}
	\begin{split}\label{eqn:fgi_gamma}
		\gamma_{\textrm{FGI}}(\kappa_c, \kappa_q;z) 
		 &= 
		 (1-z)
		 \left(
			 1- H(\kappa_c)
			 +
			\kappa_c 
		 \right) 
		 + 
		 z\left(
			 \frac{1- H(\kappa_q)}2
			+  \kappa_q 
		 \right). 
	\end{split}
\end{align}
In other words, the runtime rate function is a convex combination of the ones for the classical Sch\"oning process and for the GI scheme, with weights $(1-z), z$ respectively.
Because the classical part does not affect the coherence time, we may set $\kappa_c$ to its optimal value $\kappa_c^*=1/3$
(c.f.\ Eq.~(\ref{eqn:schoening_optimal})).
Geometrically, as we vary $z\in[0,1]$, Eq.~(\ref{eqn:fgi_gamma}) describes a line connection $(\chi_{\mathrm{GI}}(\kappa_q), \gamma_{\mathrm{GI}}(\kappa_q))$ with the parameters of the classical Sch\"oning process $(0,\gamma_{\mathrm{C}})$.
By the convexity of the GI curve, the fractional algorithm will have a better runtime rate to the left of the value of $\kappa_q$ at which the line becomes tangent to the curve.
In other words, the critical $\kappa_q$ is defined by the condition
\begin{align*}
	\frac{\partial \gamma_\mathrm{GI}}{\partial\chi} = \frac{\gamma_\mathrm{GI}-\gamma_{\mathrm{C}}}{\chi}.
\end{align*}
By a computer calculation \cite{our-data}, this happens for $\kappa_q=\frac13$ (i.e.\ equal to $\kappa_c$), resulting in the following curve:
\begin{figure}[H]
\begin{center}
	\includegraphics[scale=.75]{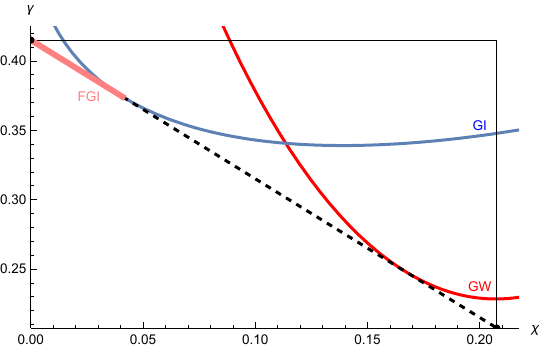}
   \end{center}
      \caption{The runtime rate vs coherence time rate for the FGI algorithm. This fractional scheme's performance is the convex combination of the classical point $(0,\gamma_\mathrm{C})$, and GI at the tangent point to the theoretical lower bound. One can note that the FGI partially saturates the optimal performance relation.}
\end{figure}

\subsubsection{Fractional Groverized Walk, Algorithm~\ref{alg:fgw}}

In the FGW scheme, we assume that the classical and Groverized walks decrease the Hamming distance in exactly $\nu_cm_c$ and $\nu_qm_q$ steps, respectively, where we have used a subscript to differentiate between the classical and Groverized random walks. 
The probabilities of such walks occurring is given by:
\begin{align}
    P(E^c_2) &= \binom{m_c}{\nu_c m_c}\Big(\frac13\Big)^{\nu_c m_c}\Big(\frac23\Big)^{(1-\nu_c)m_c},
    &
     P(E^q_2) &= \binom{m_q}{\nu_q m_q}\Big(\frac13\Big)^{\nu_q m_q}\Big(\frac23\Big)^{(1-\nu_q)m_q}
\end{align}
Analogous to the discussion in Sec.~\ref{sec:gi}, the parameters $N_2^c, N_2^q$ are defined as the reciprocals of the probabilities $P(E^c_2),P(E^q_2)$, times a constant that influences the probability failure, but will not be discussed as it has no impact on the asymptotic rates.
We further parameterize the walk lengths as $m_c = \mu_cn$ and $m_q = \mu_qn$. 
The runtime rate is
\begin{align}
    \gamma_{\textrm{FGW}}(\kappa,\nu_c,\mu_c,\nu_q,\mu_q) &= \lim_{n \rightarrow \infty} \frac1n \log\left(N_1(\kappa)N_2^c(\nu_c,\mu_c)\sqrt{N_2^q(\nu_q,\mu_q)}\right) \nonumber\\
    &= 
		1- H(\kappa) + \mu_cD(\nu_c\parallel 1/3) + 
		\frac{\mu_q}2D(\nu_q\parallel 1/3)
\end{align}
with parameters subject to the constraints

\begin{equation}
\label{constsFGW}
    \begin{aligned}
        0 &\le \kappa \le 1, \\
        0 &\le \mu_c,\mu_q, \\
        0 &\le \nu_c,\nu_q \le 1,\\
        \kappa &\leq (2\nu_c-1)\mu_c + (2\nu_q-1)\mu_q.
    \end{aligned}
\end{equation}

The first steps of the analysis should now be familiar.
There is no loss of generality in assuming that the final inequality is tight, which can be re-arranged to give

\begin{align*}
	\nu_q
	=
	\frac{\kappa - (2\nu_c-1)\mu_c}{2\mu_q} + \frac12.
\end{align*}
The rate $\gamma_{\mathrm{FGW}}$ is stationary as a function of $\mu_q$ if
\begin{align*}
	\mu_q = 3 (\kappa - (2\nu_c-1)\mu_c )
	\qquad\Rightarrow\qquad
	\nu_q = \frac23,
	\quad
	\frac{\mu_q}2D(\nu_q\parallel 1/3)
	=
	1/2 (\kappa - (2\nu_c-1)\mu_c ) 
	=
	\chi(\kappa,\nu_c,\mu_c).
\end{align*}

Eliminating $\kappa$ in favor of the coherence rate $\chi$ gives
\begin{align*}
	\kappa = 2\chi + (2\nu_c-1)\mu_c
\end{align*}

and thus
\begin{align*}
	\mu_q = 
	6\chi,
	\qquad
	\gamma_{\textrm{FGW}}(\nu_c,\mu_c;\chi) 
	&= 
	1- H( 2 \chi + (2\nu_c-1)\mu_c )
	+ 
	\mu_c D(\nu_c\parallel 1/3)
	+
	\chi.
\end{align*}
We now need to minimize $\gamma_{\textrm{FGW}}$ for fixed $\chi$ as a function of $\mu_c, \nu_c$, subject to
\begin{align*}
	0&\leq
	2\chi + (2\nu_c-1)\mu_c \leq 1, \\
	0 &\leq \mu_c, \\
	0 &\leq \nu_c \leq 1. 
\end{align*}

We may assume that $\mu_c\neq 0$, for else we are just replicating the GW scheme.
A computer calculation \cite{our-data} gives
\begin{align*}
	\partial_{\mu_c} 
	(\gamma \ln 2)
	+ 
	\frac{2-4\nu_c}{4\mu_c} 
	\partial_{\nu_c} 
	(\gamma \ln 2)
	=
	-\arctan(1-2\nu_c) + \ln(3-3\nu_c) - \frac12\ln2,
\end{align*}
which has zeros at $\nu_c=\frac13$ and $\nu_c=\frac23$.

For $\nu_c=\frac13$, one finds
\begin{align*}
	\partial_{\mu_c}
	(\gamma \ln 2)
	=
	\frac23 \arctan(1-4\chi+2/3 \mu_c)
\end{align*}
which has one zero, at $\mu_c = \frac32(4\chi-1)$.
The constraint $\mu_c\geq 0$ then implies $\chi\geq \frac14$.
But this is larger than the coherence time rate $\gamma_{\mathrm{C}}/2\simeq 0.208$ sufficient to implement a completely Groverized version of Sch\"oning's process, so this solution is not of interest.

We turn to the other solution, $\nu_c=\frac23$.
For it,
\begin{align*}
	\partial_{\mu_c}
	(\gamma \ln 2)
	=
	1/3 (-2 \arctanh(1 - 4 \chi - (2 \mu_c)/3) + \ln(2)),
\end{align*}
which has one zero:
\begin{align*}
	\mu_c=1-6\chi
	\qquad
	\Rightarrow
	\qquad
	\mu_q=6\chi,
	\quad
	\nu_c=\nu_q=\frac23,
	\quad
	\gamma_{\mathrm{FGW}} = \gamma_{\mathrm{C}} - \chi.
\end{align*}

The runtime vs coherence rate curve for the FGW scheme is given in the following figure:
\begin{figure}[H]
\begin{center}
	\includegraphics[scale=.75]{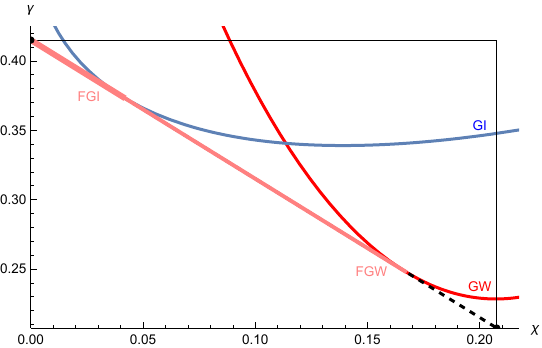}
\end{center}
\caption{The runtime rate vs coherence time rate for the FGW algorithm. This fractional scheme's performance connects the GW curve to the classical Sch\"oning point and is tangent to the curve. It achieves the optimal performance relation partially for a larger regime than FGI and for low coherence times, it comes to lie on top of the FGI line.}
\end{figure}

\subsubsection{Evenly Fractionalized Grover}

The runtime rate is
\begin{align}
	\gamma_{\textrm{EFG}}
    &= 
		(1-z)\Big(
			1- H(\kappa_c)
			+
			\mu_c D(\nu_c \parallel 1/3)
		\Big) 
		+
		z/2\Big(
			1- H(\kappa_q)
			+
			\mu_q D(\nu_q \parallel 1/3)
		\Big)
\end{align}
with success criterion
\begin{align*}
	(1-z)\kappa_c + z\kappa_q 
	=
	(1-z)(2\nu_c-1)\mu_c + 
	z(2\nu_q-1)\mu_q,
\end{align*}
which is in particular true if the following two equations hold
\begin{align*}
	\kappa_c 
	=
	(2\nu_c-1)\mu_c,
	\qquad
	\kappa_q
	=
	(2\nu_q-1)\mu_q.
\end{align*}
But this is just the convex interpolation between a completely classical and a completely Groverized process.
In particular, by choosing the parameters as for the original Sch\"oning process
\begin{align*}
	\nu_c=\nu_q=\frac23,
	\quad
	\kappa_c=\kappa_q=\frac13,
	\quad 
	\mu_c=\mu_q=1,
\end{align*}
we obtain a coherence time--runtime rate curve that linearly connects the classical point $(0,\gamma_{\mathrm{C}})$ to the completely Groverized one $(\gamma_{\mathrm{C}}/2, \gamma_{\mathrm{C}}/2)$ (Fig.~\ref{fig:all_tradeoffs}).

\begin{figure}[ht]
    \begin{center}
			\includegraphics[scale=.75]{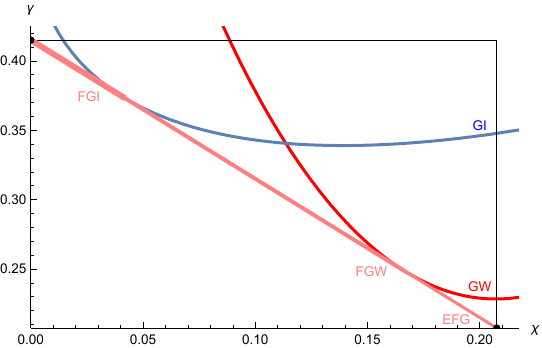}
    \end{center}
		\caption{\label{fig:all_tradeoffs} 
			Runtime--coherence time rate curves for the covered algorithms.
			The linear interpolation between the classical and the completely Groverized points are realizable using an increasing number of  methods -- first only EFG, then also FGW, finally also FGI -- as the coherence time decreases.
		}
\end{figure}

\subsection{A heuristic de-randomization of the GI schemes}

In this section, we provide evidence that the Groverized initialization schemes can reach further into the  $\gamma$--$\chi$ chart than what the Markovian model suggests.
To see why this is plausible,
note that the role of randomness for the initial configuration $x$ is very different from the role of randomness for the walk decisions $w$.
In the first case, there is an ``absolute measures of the quality of the initial configuration'', namely the Hamming distance to the solution.
The probability that the walk does find the solution is quite obviously a function of that metric.
Therefore, baring major algorithmic insights, it is unavoidable to consider many different initial configurations before encountering  one that will likely lead to a solution.

In contrast, it is not implausible that ``every walk works for equally many initial configurations'', i.e.\ that there are no choices for $w$ that are ``intrinsically better than others''.
More precisely, it seems reasonable to assume that for 
sufficiently large $n$, and
generic SAT formulas, it holds that with high probability in $w$
\begin{align}\label{eqn:w_doesnt_matter}
	&-\frac1{n} \log \left(
	\Pr_x[ \textsc{SchoeningWalk}(x, w) = x^\star \,|\, d_H(x,x^\star)=h,w] \right) \\ \nonumber
	&\simeq
	-\frac1{n} \log\left(
		\Pr_{x,w'}[ \textsc{SchoeningWalk}(x, w') = x^\star \,|\, d_H(x,x^\star)=h]
	\right).
\end{align}
The right hand side can be easily calculated, as by Ref.~\cite{Schoening99}, for $\mu=3$, 
\begin{align*}
	\Pr_{x,w}[ \textsc{SchoeningWalk}(x, w) = x^\star \,|\, d_H(x,x^\star)=h]
	=
	2^{-h}.
\end{align*}

Under Assumption (\ref{eqn:w_doesnt_matter}), one can restrict the outer loop over $w$'s from Alg.~\ref{alg:gi} to $N_2=1$ iteration, and compensate by increasing the number of Grover iterations for $x$ to $N_1=O^*(2^{\gamma_{\mathrm{C}}/2n})$.
In other words, the Groverized Initialization scheme with these parameters would lie on the optimal point 
$(\chi,\gamma) = (\gamma_{\mathrm{C}}/2, \gamma_{\mathrm{C}}/2)$.

Being even bolder, one could then speculate that the analysis of Sec.~\ref{sec:fgi} carries 
over and that, as one varies the fraction of initialization bits that are subjected to a Grover search, one could trace out the optimal $(\chi,\gamma)$-line.
In other words, it does not seem impossible that the following Alg.~\ref{alg:hfgi}, with parameter choice
\begin{align*}
	N_1^{(c)} = O^*(2^{\gamma_{\mathrm{C}}(1-z)n}),
	\qquad
	N_1^{(q)} = O^*(2^{\gamma_{\mathrm{C}}zn/2}),
\end{align*}
achieves the optimal trade-off.

\begin{algorithm}[H]
	\caption{Heuristically De-Randomized Fractional Groverized Initialization}\label{alg:hfgi}
	\begin{algorithmic}[1]
			\State $w\leftarrow$ uniformly random value from $\{1, 2, 3\}^{\times m}$
				\For{$j=1 ...  N_1^{(c)}$}
				\State $x_c\leftarrow$ uniformly random value from $\{0,1\}^{\times \lceil(1-z)n\rceil}$
				\State $x_q\leftarrow$ Grover-search for $\Big\lfloor \sqrt{N_1^{(q)}}\Big\rfloor$ iterations using \textsc{Oracle$_{(x_c, w)}$}()
				\State $x=(x_c,x_q)$
					\If{$x$ satisfies all clauses} 
						\State\Return $x$
					\EndIf
				\EndFor
			\State \Return False
	\end{algorithmic}
\end{algorithm}

To gather evidence in favor of Assumption~(\ref{eqn:w_doesnt_matter}), we have resorted to numerical methods.
A first ansatz is to compute the l.h.s.\ of Eq.~(\ref{eqn:w_doesnt_matter}) exactly, which is possible for small values of $n$ by iterating over all $2^n$ assignments to $x$.
Results are shown in Fig.~\ref{fig:numerics_20_b} for a randomly chosen set of $3$-SAT formulas with $n=20$ variables, $L=91$ clauses. 
The number of satisfying assignments $t_0$ of the formulas are varied.
Only the case $t_0=1$ can be directly compared to the analytic bounds.
However, note that even for this case, the empirically observed rate of $\gamma_{\mathrm{GI}}\simeq .12 \pm .02$ is much lower than the value $\gamma_{\mathrm{C}}/2 \simeq .208$ that we would expect theoretically.
Presumably, $n=20$ is still too small to show the asymptotic behavior.

\begin{figure}[ht]
	\centering
        \includegraphics[width=.45\textwidth]{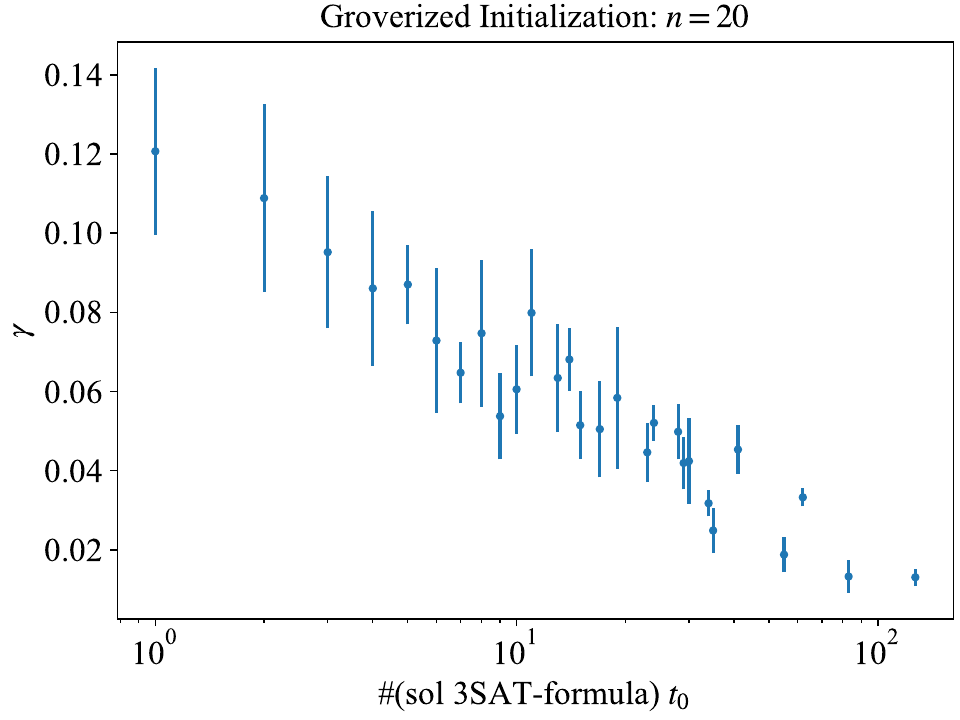}
        \caption{\label{fig:numerics_20_b}
					Plot of the runtime rate for the heuristically de-randomzied GI scheme.
					Error bars indicate variation as a function of the formulas and the walk variables $w$.
					On the $x$-axis, we show the number of satisfying assignments in the formula. 
					Only the case of $t_0=1$ should be directly comparable to the analytic bounds.
					The empirically observed behavior is much better than the analytic results, suggesting that $n=20$ is too small to capture the asymptotic behavior.
			}
\end{figure}

To test this assumption, we had to turn to numerical heuristics, to at least probe the behavior for much larger values of $n$, where an exact computation is no longer possible.
The results are shown in Fig.~\ref{fig:runtime_1414}.
We used a SAT instance with $n=1414$ variables that we believe to have a single satisfying assignment $x^\star$ which is explicitly known.
To generate the instance, a 128-bit plain text was encoded by a 128-bit key using the XTEA block cipher truncated to three rounds.
The formula represents the conditions on an input key to map the known plain text to the known ciphertext. 
The clauses are designed such that they enforce the correct evaluation of bit-wise operations of the algorithm with respect to the given input and output. 
XTEA was restricted to three rounds in order to keep the size of the formula manageable. 
While we have no formal proof, it is reasonable to assume that there is a unique key that satisfies the formula.
This is supported by consistency checks in terms of running SAT solvers on a version of this problems with even fewer rounds \cite{philipp-private-commutation}.

Let us denote the sphere of strings with Hamming distance $h$ from $x^\star$ by $M^h(x^\star)$. For a fixed walk randomness $w$, and for $h=1, \dots 11$, we have drawn $x$ uniformly from $M^h(x^\star)$.
In order to compare the numerical results to the theory prediction, we have to use the value of the right hand side of Assumption~(\ref{eqn:w_doesnt_matter}) for non-asymptotic values of $n$.
The following plot shows the empirically estimated probabilities of Sch\"oning's walk (with $\mu=3$) arriving at the solution, when starting from a random initial configuration of given Hamming distance.
The findings show the expected behavior of averaging over $w$, already for a fixed random value of $w$.
In this sense, they are compatible with Assumption~(\ref{eqn:w_doesnt_matter}).
We note, however, that we were not able to probe the assumption for larger values of $h$.
Garnering a better understanding for the concentration properties of the Sch\"oning walk as a function of the walk choices remains therefore an open question.

\begin{figure}[ht]
    \centering
    \includegraphics[width=.45\textwidth]{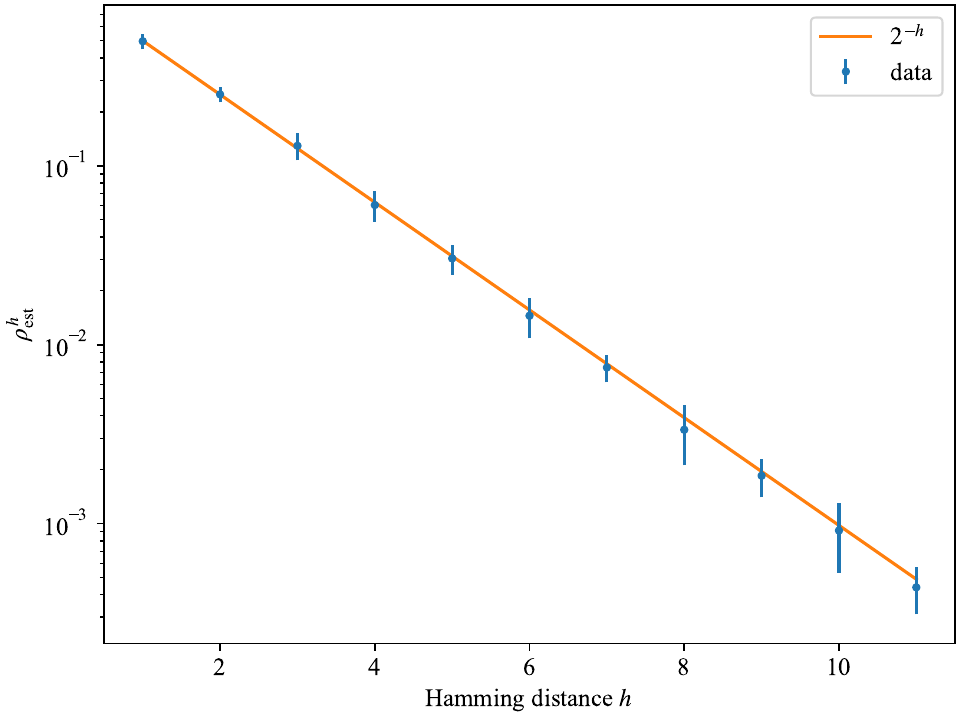}
   \caption{\label{fig:runtime_1414}
		 Estimated probability for a uniformly random initial configuration $x$ with Hamming distance $h$ to be mapped to $x^\star$ under a Sch\"oning walk, for a fixed, randomly chosen set of walk decisions $w$ (c.f.\ Alg.~\ref{alg:hfgi}).
		 The SAT instance has $n=1414$ variables and is believed to have a unique satisfying assignment \cite{philipp-private-commutation}.
		 For each data point, $10^4$ initial configurations $x$, were sampled uniformly from the Hamming distance sphere $M^h\left(x^\star\right)$. 
 The results agree well with the theoretical prediction under Assumption~(\ref{eqn:w_doesnt_matter}) (orange line). 
 }
\end{figure}

\section{Circuits\label{circuit}}

In this section, we discuss an implementation of the partial Groverization schemes and present the main building blocks of their quantum circuits. Given $n$ variables and the length of Sch\"oning's walk $m$, the quantum implementation requires $n + m \log3$ qubits to encode the initializations and walk randomness. 
The oracles of the partial Groverization schemes are some adaptation of one or more Sch\"oning walks, and regardless of the search space they act on, the label of the violated clause at each step needs to be stored in their workspaces. This is necessary since such oracles are typically realized using uncomputation, therefore, $\log L$ extra auxiliary qubits are needed at each step, amounting to $m\log L$ qubits in total for the workspace. As a result, encoding any Groverization of Sch\"oning's algorithm asymptotically needs $n+(\log3+\log L)m$ qubits. 

Figure \ref{fig:general_circuit} represents a single step of Sch\"oning walk, schematically. The first register encodes the space of all possible initialization. The gates $\text{ev}_j$, for $j \in\{1,..,L\}$, act on the first two registers. Each gate consists of a few controlled-gates where the control qubits correspond to the three variables in the $j$-th clause, and the target qubit is the second register. The second register is an auxiliary qubit, initially set to $\ket{0}$, and is negated as soon as the first violated clause is detected. The third register consists of $\log L$ auxiliary qubits that are used to count the number of clauses from where the first violated clause has happened. The last register is a qutrit providing the randomness of the corresponding walk step. The controlled-gates $\text{ch}_j$, for $j \in\{1,..,L\}$ act on the first three registers, and take care of variable flipping wherever the first violated clause is detected. The $0\lor1\lor2$ block represents a triple controlled-gate where the control qutrit is the subspaces corresponding to the computational basis states $\ket{0},\ket{1},\ket{2}$. 
Figure \ref{fig:GI_circ} depicts the controlled-gates including $\text{ch}_j$, in detail. The sub-figure on the right shows the corresponding controlled-gate for GI, where the walk randomness is fed classically to the last register.

All partial Groverization of Sch\"oning algorithm can be implemented using slight modifications.
For the GW algorithm, the $n$-qubit variable register will not be initialized in the uniform superposition of all possible assignments $\ket{+}^{\otimes n}$, but rather in a state with classically randomly defined variables $\ket{x_1 \cdots x_n}$.
For the GI algorithm the qutrit within every Sch\"oning's step can be removed since we can, for every Sch\"oning's step, generate a random number $r\in\{0, 1, 2\}$ and apply only the $X$ gates based on the classically determined $r$ (see figure \ref{fig:GI_circ}).

\begin{figure}[t]
    \begin{center}
        \begin{minipage}{\textwidth}
            \resizebox{\textwidth}{!}{%
\begin{quantikz}
\lstick{$\ket{+}^{\otimes n}$} &
    [3mm] \qwbundle{n} & 
    \gate[2]{\text{ev}_1}  & 
    \gate{\text{ch}_1} & 
    \qw & 
    \gate[2]{\text{ev}_2} & 
    \gate{\text{ch}_2} & 
    \qw & 
    \qw & 
    &
    &
    \gate[2]{\text{ev}_L} & 
    \gate{\text{ch}_L} & 
    \qw \\ 
\lstick{$\ket{0}$} & 
    \qw & 
    \qw & 
    \ctrl{-1} & 
    \ctrl{1} & 
    \qw & 
    \ctrl{-1} & 
    \ctrl{1} & 
    \qw & 
    \cdots &
    & 
    \qw & 
    \ctrl{-1} & 
    \qw \\ 
\lstick{$\ket{0}^{\otimes \log L}$} & 
    \qwbundle{\log L} & 
    \qw & 
    \octrl{-1} & 
    \gate{+1} & 
    \octrl{-1} & 
    \octrl{-1} & 
    \gate{+1} & 
    \qw & 
    &
    &
    \octrl{-1} & 
    \octrl{-1} & 
    \qw \\ 
\lstick{
}  & 
    \qw & 
    \qw & 
    \gate{0\lor1\lor2}\vqw{-1} & 
    \qw & 
    \qw & 
    \gate{0\lor1\lor2}\vqw{-1} & 
    \qw & 
    \qw & 
    &
    &
    \qw & 
    \gate{0\lor1\lor2}\vqw{-1} & 
    \qw\\ 
\end{quantikz}

            }
        \end{minipage}
    \end{center}
    \caption{\label{fig:general_circuit} The quantum implementation of a single Sch\"oning's step for a general implementation of the partial Groverization of
 Sch\"oning's algorithm. The $\text{ev}_j$ gates evaluate the $j$-th clause on the corresponding variables and the controlled-gates containing $\text{ch}_i$ and  $0\lor1\lor2$ act on all the registers and check if the $j$-th clause is the first violated clause and if so, flip one of three variables in it based on the randomness provided by the if-statement, $0\lor1\lor2$. Here $0\lor1\lor2$ represents a triple controlled-gate where the control qutrit is the subspaces of the computational basis (visualized in figure \ref{fig:GI_circ}). The $\log L$ auxiliary qubits are needed for uncomputation.}
\end{figure}
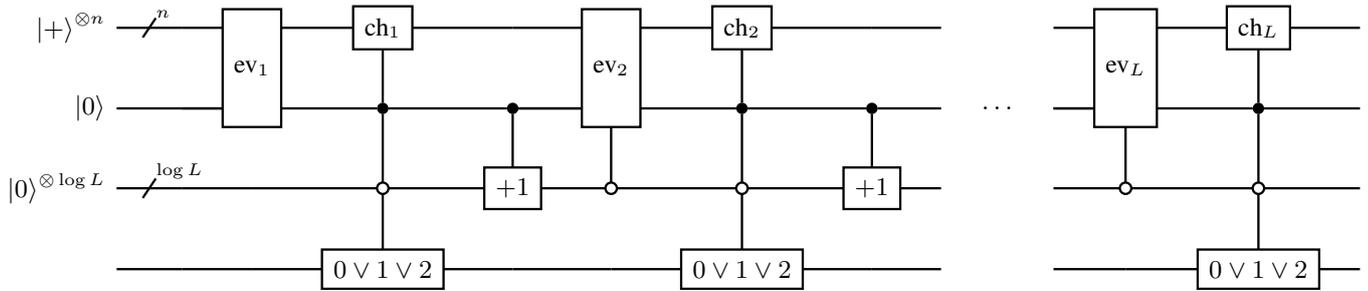
   
\begin{figure}[t]
\begin{minipage}{\textwidth}
        \resizebox{\textwidth}{!}{%
       \begin{quantikz}
\lstick{$\ket{x_{i_1}}$}    & \gate[3]{\text{ch}_j} & \qw &\\ [3.3mm]
\lstick{$\ket{x_{i_2}}$}    & \qw          & \qw & \\ [3.3mm]
\lstick{$\ket{x_{i_3}}$}    & \qw          & \qw & \\
\lstick{$\ket{\chi}$}   & \ctrl{-1}    & \qw & \\
\lstick{$\ket{\phi}$} 
                        & \octrl{-3}   &\qw & \\
\lstick{
}   
                        & \gate{0\lor1\lor2}\vqw{-1} &\qw \\
\end{quantikz}
=
\begin{quantikz} 
\lstick{$\ket{x_{i_1}}$}    & \gate{X} & \qw & \qw & \qw&\\
\lstick{$\ket{x_{i_2}}$}    & \qw & \gate{X} & \qw & \qw&\\
\lstick{$\ket{x_{i_3}}$}    & \qw & \qw & \gate{X} & \qw&\\
\lstick{$\ket{\chi}$}   & \ctrl{-3} & \ctrl{-2} & \ctrl{-1} & \qw &\\
\lstick{$\ket{\phi}$} 
                        & \octrl{-1} & \octrl{-1} & \octrl{-1}  &\qw & \\
\lstick{
}  
                        & \gate{0}\vqw{-1} & \gate{1}\vqw{-1} & \gate{2}\vqw{-1} & \qw & \\
\end{quantikz} $\xRightarrow{\text{GI}}$ \begin{quantikz}
\lstick{$\ket{x_{i_1}}$}    & \gate{X} & \qw & \qw & \qw\\
\lstick{$\ket{x_{i_2}}$}    & \qw & \gate{X} & \qw & \qw\\
\lstick{$\ket{x_{i_3}}$}    & \qw & \qw & \gate{X} & \qw\\
\lstick{$\ket{\chi}$}   & \ctrl{-3} & \ctrl{-2} & \ctrl{-1} & \qw&\\
\lstick{$\ket{\phi}$} 
                        & \octrl{-3} & \octrl{-2} &\octrl{-1}  &\qw\\
\lstick{$r$}            & \gate{0}\vcw{-1} & \gate{1}\vcw{-1} & \gate{2}\vcw{-1}&\qw\\[-.66cm]
                        \lstick{} &\qw & \qw & \qw & \qw &\\
\end{quantikz}
        }
\end{minipage}
\caption{\label{fig:GI_circ} Implementation of the variable flips of Sch\"oning's walk within amplitude amplification.
Here, $x_{i_1}$, $x_{i_2}$ and $x_{i_3}$ are the variables of the $j$-th clause.  The $0\lor1\lor2$ block represents a triple controlled-gate where the control qutrit is the  subspaces of the basis $\ket{0},\ket{1},\ket{2}$. 
For the GI algorithm, the walk randomness can be provided by fixing a random number $r \in \{0,1,2\}$ for every walk step. }
\end{figure}
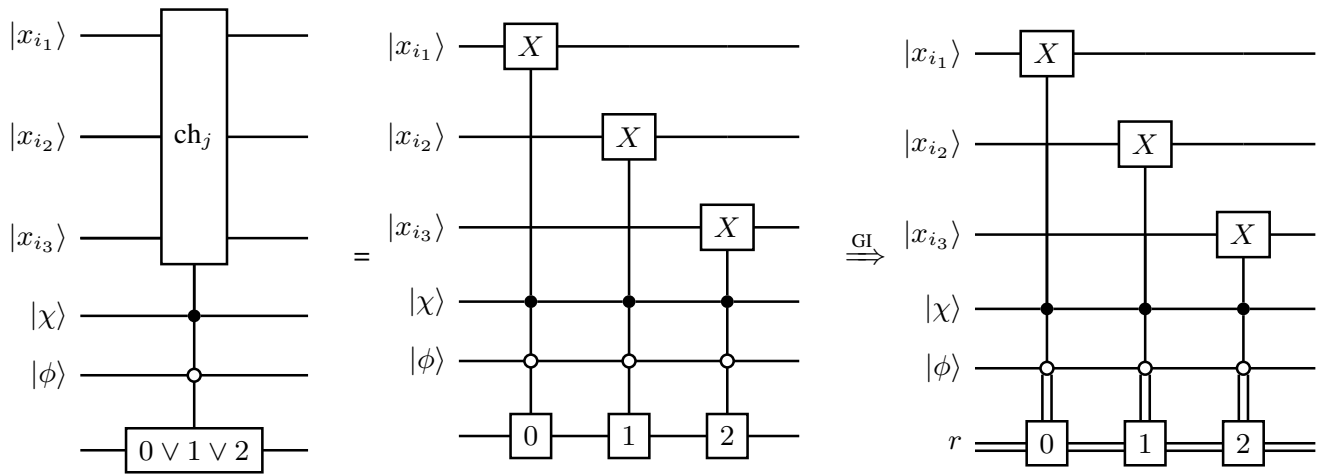

\section{Summary \& Outlook}

This work considers hybrid schemes for search-based quantum algorithms, with the aim to allow for parallelizability, and to reduce the need for long coherence times. The basic gist is to partition the randomness of an underlying classical probabilistic algorithm into a part that is subject to Grover search, while the rest is sampled classically. Such `partial Groverizations' allow for parallelization of the classical sampling, as well as enable adaption to available coherence times. We consider exponential-time algorithms, why our analysis focuses on the asymptotic run-time rates and coherence-time rates. We argue that these two types of rates are bounded by a general trade-off relation that no hybrid-scheme can beat. For our concrete analysis, we consider hybrid schemes based on  Sch\"oning's algorithm, where the latter solves $3$-SAT (or more generally $k$-SAT) problems by random walks in the space of assignments. The walk-procedure allows for several partial Groverization-schemes.  We determine the corresponding run-times and coherence-times of these schemes, and demonstrate saturation of the general trade-off relation.  Many of these partial Groverizations intuitively lend themselves for efficient circuit implementations, and we provide the main building blocks of these. On a more speculative note, we present numerical evidence that the GI scheme can be partially de-randomized, in the sense that a single `typical' instance of the classical randomness of the walk appears to mimic the effects of the repeated sampling. This would open for an additional flexibility in the implementation of these hybrid-schemes, still maintaining the optional trade-off.

In this investigation, we have focused on partial Groverizations of  Sch\"oning's algorithm. However, this approach should in principle be applicable to any classical probabilistic search scheme, since it essentially only rests on partitions of the underlying randomness. The main concern would be to find `natural' partitions that are algorithmically accessible, in the sense that the partial Groverization can be implemented efficiently.  Explicit run-time and coherence-time rates would also require a classical scheme, as well as partitions, that are sufficiently tractable for analysis, unless one would resort to numerical estimates.

The partial de-randomization of GI-scheme that is suggested by our numerical explorations, would deserve further investigations.    
In particular, the question is to what extent, and in what sense,  the hypothetical relation \eqref{eqn:w_doesnt_matter} would be true. Moreover, one may ask if something similar also would apply to fractional GI. For numerical investigations, it would be relevant to extend to larger Hamming distances, further classes of $3$-SAT  instances, as well as problem sizes. This would likely involve challenges to design reliable numerical estimates, since exact calculations by the very nature of the problem quickly becomes intractable. For purely analytical approaches, some notion of concentration of measure of walks, would be interesting.

In the spirit of \cite{Schoening99, SchoeningToranBook} we have in this investigation employed  `the walk on $\mathbb{Z}$' as a model of the true Sch\"oning-procedure. In Appendix (see also \cite{SwissPhDThesis}) we in additionally provide bounds for the true rates of Sch\"oning-procedure and the GW-procedure, in terms of the mirroring processes on $\mathbb{Z}$. It would be relevant to obtain similar bounds also for the GI-process, as well as for the various fractional schemes.

\section{Acknowledgement}

We thank Phillip Keldenich for kindly providing us with the SAT instance that was used for numerical simulations resulted in figure (\ref{fig:runtime_1414}). 
We acknowledge support from Bundesministerium für Bildung und Forschung - BMBF under project QuBRA.
The UzK team was also supported by Germany’s Excellence Strategy – Cluster of Excellence Matter and Light for Quantum Computing (ML4Q) EXC 2004/1 (390534769).



\newpage
\appendix

\section{\label{AppFromSchToZ}From the true Sch\"oning-process  to the Markov process on $\mathbb{Z}$}

\subsection{The purpose of this appendix}

For the calculation of rates, we replace the genuine searches of solutions for 3SAT-problems (the `true  Sch\"oning process'), with a Markovian random walk on  the `Hamming distance' (although we strictly speaking consider a walk on $\mathbb{Z}$). This is analogous to Sch\"onings analysis of the performance of Sch\"onings algorithm \cite{Schoening99,SchoeningToranBook}, where it is argued that this substitute-process yields an upper bound on the rates of the run-time of the algorithm. The purpose of this appendix is to give a more detailed argument for why the success probability of Sch\"onings algorithm is lower bounded by the success-probability of the substitute walk on $\mathbb{Z}$. The reader may also wish to consult \cite{SwissPhDThesis} for a previous analysis along these lines. Apart from obtaining from bounding the success probability for the true Sch\"oning-process, we also provide the analogous bound for the GW process.


\subsection{The Sch\"oning-process}

As described in the main text, the 3SAT problem consists of  a collection of clauses $C_1,\ldots, C_L$ on $n$ binary variables, where each clause is of the form $C_j = l_0^{(j)}\vee l_1^{(j)}\vee l_2^{(j)}$, and where each of the literals $l_0^{(j)}$,  $l_1^{(j)}$, $l_2^{(j)}$ is one of the binary variables, or its negation.   The  3SAT formula is the conjunction of all the given clauses, $C := \wedge_{j=1}^{L}C_j$, and the task is to determine whether there exists an assignment $x\in \{0,1\}^{\times n}$ of the $n$ binary variables, which satisfies $C$.  In the following analysis we assume that $C$ either has a \emph{unique} satisfying assignment $x^{\star}\in \{0,1\}^{\times n}$, or alternatively, that $x^{\star}$ is \emph{selected} among a set of solutions.

 Sch\"onings procedure can be regarded as a stochastic process $(x_l)_{l=0}^{m}$ with  $x_l \in\{0,1\}^{\times n}$. 
The process is initialized by a random assignment $x_0$ of the bit string, drawn uniformly over all of $\{0,1\}^{\times n}$. On this state it checks all the clauses $C_1,\ldots,C_L$ (according to a pre-determined order). If all are satisfied, then the initial string satisfies $C$ and the algorithm terminates. 
 Otherwise, it finds the first unsatisfied clause, and randomly negates one of the three variables corresponding to the literals of that clause. The algorithm continues according to this random walk until it either finds a satisfying assignment, or it reaches  a  pre-determined termination-time $K$.  For our purposes, it is convenient to think of the state $x_l$ of the process as a function of a collection of random variables. The initialization is represented by the random variable $A$, which takes values in $\{0,1\}^{\times n}$. The randomness in the walk is captured by the variables  $B = (B_1,\ldots, B_m)$ be random variables where each $B_l$ takes values in $\{0,1,2\}$ (and thus $B$ takes values in $\{0,1,2\}^{\times m}$). Hence, $B_l$ represents one of the three possible choices of which literal to flip at step $l$. 
   We assume that $A,B_1,\cdots, B_m$ are independent and uniformly distributed, i.e., for $b = (b_1,\ldots, b_m)$ we have
\begin{equation}
\label{dfabadfbadfn}
\begin{split}
P(A= a,B=b) = & P(A=a)P(B=b) = P(A = a)P(B_1 = b_1)\cdots P(B_m = b_m),\\
P(A = a) = & \frac{1}{2^n},\quad \forall a\in\{0,1\}^{\times n}\\
 P(B_l=b_l) = & \frac{1}{3},\quad \forall b_l\in{0,1,2}.
\end{split}
\end{equation}
Hence, we can write the Sch\"oning process as $(x_l)_{l} = (x_l(A,B))_{l}$, where
\begin{equation}
x_0(a,b) := a,
\end{equation}
i.e., $a$ the initial state. At the  l:th step of Sch\"oning's process is based on the state $x_{l-1}$ of the previous step. On this state, all the clauses $C_1,\ldots,C_L$ (according to a pre-determined order) are checked. If all are satisfied, then $x_{l-1} = x^{\star}$ and the process remains in that state, i.e., $x_l = x^{\star}$.
(In other words, t $x^{\star}$ is an absorbing state for the Sch\"oning-process.)
 Otherwise, it finds the first unsatisfied clause, which we refer to as $C_{j_l}$. The selected clause, $C_{j_l}$, contains the three literals $(l_0^{(j)},l_1^{(j)},l_2^{(j)})$. The process constructs $x_l$ by negating the variable corresponding to literal $l^{(j)}_{b_l}$. In other words, it is the l:th component of $b$ that determines which of these three choices that is selected. One may note that the process, by construction, satisfies
\begin{equation}
x_l(a,b) = x_l(a,b_1,\ldots,b_{l}).
\end{equation}
Hence, the value of $x_l(a,b)$ only depends on the values of $b_1,\ldots, b_l$, not any of the `later' variables $b_{l+1},b_{l+1},\ldots$.
One may also note that $A,B_1,B_2,\ldots, B_K$ encompasses \emph{all} the randomness in the process. In other words, the state $x_l$ is uniquely determined by $a,b_1,\ldots,b_{l}$.


\subsection{The proof idea}

As described above, the true Sch\"oning-process $(x_l)_l$ is a walk on bit-strings. However, for the analysis of the optimal rates, we follow in the steps of  Sch\"oning \cite{Schoening99,SchoeningToranBook}, and instead focus on the Hamming-distance to the (selected) solution $x^{\star}$. In principle, nothing prevents us from projecting the state $x_l$ of the  Sch\"oning-process, to the Hamming distance $d_H(x_l,x^{\star})$ (i.e. projecting onto $\mathbb{N}$). However, this would generally yield a  process that would be no easier to analyze than the original Sch\"oning-process. One may for example note that although Sch\"onings-process  $(x_l)_l$ is Markovian on the space of bit-strings, one cannot generally expect its projection $\big(d_H(x_l,x^{\star})\big)_l$ to be Markovian on $\mathbb{N}$. The general idea for the analysis is to replace (via  a coupling) the true projection   $\big(d_H(x_l,x^{\star})\big)_l$ with another process $(\tilde{d}_l)_l$ on $\mathbb{N}$, which is Markovian and which moreover upper-bounds the true Hamming-distance, $d_H(x_l,x^{\star})\leq \tilde{d}_l$. One may note that the Sch\"oning-process is `successful' if it finds the solution $x^{\star}$. Hence,  we can express the success probability at step $l$ as $P(x_l = x^{\star}) = P\big(d_H(x_l,x^{\star}) = 0\big)$. From the bound $d_H(x_l,x^{\star})\leq \tilde{d}_l$ it follows that 
$P(x_l = x^{\star}) \geq  P(\tilde{d}_l = 0)$. In other words, the success-probability of the Sch\"oning-process is lower-bounded by the probability that the substitute-process $\tilde{d}_l$ reaches $0$. The fact that $(\tilde{d}_l)_l$ is Markovian makes the analysis more tractable. However, the value $0$ corresponds to an absorbing boundary. (If we find the solution at an earlier stage, we should terminate the process rather than walking on.) To further ease the analysis, we remove this boundary and instead introduce yet another walk $(d_l)_l$ on $\mathbb{Z}$, which we regard as `successful' whenever $d_l\leq 0$. For this process we moreover establish the bound $P(\tilde{d}_l = 0)\geq P(d_l\leq 0)$, and thus $P(x_l = x^{\star})\geq P(d_l\leq 0)$. By the trivial bound $P(d_l\leq 0)\geq P(d_l = 0)$, we thus ultimately get the bound  $P(x_l = x^{\star})\geq P(d_l =  0)$. For the calculation of the optimal rates, our starting point is an expression for $P(d_l =  0)$. By the inequality $P(x_l = x^{\star})\geq P(d_l =  0)$ it follows that the calculated rates are upper bounds to the true rates of the Sch\"oning-process.


\subsection{Constructing a walk $(\tilde{d}_l)_l$ on  $\mathbb{N}$ such that $d_H(x_l,x^{\star})\leq \tilde{d}_l$}

Related to the Sch\"oning-process $(x_l)_l$, we here wish to construct  another process $(\tilde{d}_l)_l$, where $\tilde{d}_l$ takes values in $\mathbb{N}$ for all $l\in\mathbb{N}$, and is such that 
   \begin{equation}
d_H\big(x_l(a,b_{1},\ldots, b_{l}),x^{\star}\big) \leq \tilde{d}_l(a,b_{1},\ldots, b_{l}),\quad \forall a\in\{0,1\},\quad \forall b\in\{0,1,2\}^{\times m},\quad l = 0,1,2,\ldots,m.
\end{equation}
In other words, we want to make sure that  $\tilde{d}_l$ \emph{always} is an upper bound to the Hamming distance between $x_l$ and $x^{\star}$.
This requires a considerable coordination between the two processes. In particular, whenever $x_l$ moves in the `wrong' direction (i.e, increases the Hamming distance to $x^{\star}$) then $\tilde{d}_l$ also has to increase. 
To this end, we consider the list of clauses $C_{1},\ldots, C_L$. For each clause $C_j$  it is the case that $C_j(x^{\star}) = 1$. Hence, for each $j$, at least one of the literals  $l_0^{(j)},l_1^{(j)},l_2^{(j)}$ is satisfied by $x^{\star}$. Among these satisfied clauses we select one of these satisfied literals, and let $r_j\in\{0,1,2\}$ be its index. In other words, we are guaranteed that $l_j^{(r_j)}(x^{\star}) = 1$.

As already described above, the Sch\"oning-process $(x_l(a,b))_l$ is uniquely determined by $(a,b_1,\ldots,b_l)$, and does in turn uniquely determines the unsatisfied clauses $C_{j_l}$, as long as $x_l(a,b_1,\ldots, b_l)\neq x^{\star}$. Consequently, it also uniquely determines a sequence of `selected' literals $r_{j_l}$, whenever $x_l(a,b_1,\ldots, b_l)\neq x^{\star}$.
For each $l \in 1,2,\ldots $, we define a mapping $(a,b_1,\ldots b_{l-1})\mapsto f_l(a,b_1,\ldots,b_{l-1})\in\{0,1,2\}$ by
\begin{equation}
	\label{sfgnsfgngf}
\begin{split}
f_1(a)    := & \left\{\begin{matrix}
 0 & \mathrm{if} &  x_0 \equiv a = x^{\star},\\
 r_{j_1} & \mathrm{if} &  x_{0}\equiv a \neq x^{\star}.
\end{matrix}\right.\\
f_l(a,b_1,\ldots,b_{l-1}) := & \left\{\begin{matrix}
 0 & \mathrm{if} &  x_{l-1}(a,b_1,\ldots, b_{l-1}) = x^{\star},\\
 r_{j_l} & \mathrm{if} &  x_{l-1}(a,b_1,\ldots, b_{l-1}) \neq x^{\star}.
\end{matrix}\right.\quad l = 2,3,\ldots
\end{split} 
\end{equation}
The purpose of $f_l(a,b_1,\ldots,b_{l-1})$ is to determine which value of $b_l$ that should correspond to a `successful' move for the $(\tilde{d}_l)_l$-process.
More precisely, we define $(\tilde{d}_l(a,b))_{l}$ by
\begin{equation}
\label{fgnsfgnfg}
\begin{split}
\tilde{d}_0(a,b) := &  d_H\big(x_0(a,b),x^{\star}\big) = d_H(a,x^{\star}),\\
\tilde{d}_{l}(a,b_{1},\ldots, b_{l}) := & \left\{\begin{matrix}
0 & \mathrm{if} & \tilde{d}_{l-1}(a,b_{1},\ldots,b_{l-1}) = 0\\
\tilde{d}_{l-1}(a,b_{1},\ldots, b_{l-1}) + 1 & \mathrm{if} & \tilde{d}_{l-1}(a,b_{1},\ldots, b_{l-1}) \neq 0,\quad  b_l \neq f_l(a,b_{1},\ldots, b_{l-1})\\
\tilde{d}_{l-1}(a,b_{1},\ldots, b_{l-1}) - 1 & \mathrm{if} &  \tilde{d}_{l-1}(a,b_{1},\ldots, b_{l-1}) \neq 0,\quad b_l  = f_l(a,b_{1},\ldots,b_{l-1})\\
\end{matrix}\right. \quad l = 1,2,\ldots
\end{split}
\end{equation}
In words, the first condition in the bracket means that $0$ is an absorbing state, i.e., if $\tilde{d}_l(a,b) = 0$ for some $l$, then $\tilde{d}_{l'}(a,b) = 0$ for all $l'\geq l$. The other two cases make sure that the $\tilde{d}_l$ moves in `coordination' with the  Sch\"oning-process $(x_l(a,b))_l$, in such a manner that it is guaranteed that $d_H\big(x_l(a,b_{1},\ldots,b_{l}),x^{\star}\big)$ does not increase above $\tilde{d}_l(a,b_{1},\ldots,b_{l})$.

\begin{lemma}
\label{bound}
The Sch\"oning process $(x_l)_{l\in\mathbb{N}}$ and the process $(\tilde{d}_l)_{l\in \mathbb{N}}$ as defined by  (\ref{sfgnsfgngf}) and (\ref{fgnsfgnfg}) satisfy 
  \begin{equation}
  \label{adfbdfn}
d_H\big(x_l(a,b_{1},\ldots,b_{l}),x^{\star}\big) \leq \tilde{d}_l(a,b_{1},\ldots,b_{l}),\quad \forall a\in\{0,1\}^{\times n},\quad \forall b\in\{0,1,2\}^{\times l},\quad l = 0,1,2,\ldots.
\end{equation}
\end{lemma}
One may note that (\ref{adfbdfn}) holds for every single element in the event-space, and Lemma \ref{bound} does thus not depend on the actual probability distribution of $A,B_1,\ldots, B_l$. However, there are other steps in our proofs that do depend crucially on these variables being independent and uniformly distributed.

\begin{proof}
We first note that 
\begin{equation}
\begin{split}
x_0(a,b) = a,\quad \tilde{d}_0 = d_H(a,x^{\star})
\end{split}
\end{equation}
and thus (\ref{adfbdfn}) is satisfied for $l=0$ for all $a,b$.

Now, assume that (\ref{adfbdfn}) holds for some $l-1,a,b$.
We have the following cases:
\begin{itemize}
\item {\bf Case $x_{l-1}(a,b) = x^{\star}$:}  Since we assume that $x^{\star}$ is absorbing, it follows that $x_{l}(a,b) = x^{\star}$ and consequently $d_H\big(x_l(a,b_{1},\ldots,b_{l}),x^{\star}\big) = 0$.  Concerning $\tilde{d}_{l-1}$, we can distinguish yet two sub-cases:
\begin{itemize}
\item {\bf Case $\tilde{d}_{l-1}(a,b_1,\ldots,b_{l-1}) = 0$:} By construction (first case in (\ref{fgnsfgnfg})) $\tilde{d}_{l}(a,b_1,\ldots,b_{l}) = 0$, and (\ref{adfbdfn}) is thus satisfied for $l,a,b$.
\item {\bf Case $\tilde{d}_{l-1}(a,b_1,\ldots,b_{l-1}) \neq 0$:} Then $\tilde{d}_{l-1}(a,b_1,\ldots,b_{l-1}) \geq 1$. Since the process $d$ can change at most one step, it follows that $\tilde{d}_{l-1}(a,b_1,\ldots,b_{l-1}) \geq 0$, and thus (\ref{adfbdfn}) is satisfied for $l,a,b$.
\end{itemize}
\item {\bf  Case $x_{l-1}(a,b) \neq x^{\star}$:}
Since we assume that (\ref{adfbdfn}) holds for $l-1,a,b$ it follows that  $\tilde{d}_{l-1}(a,b_1,\ldots,b_{l-1}) \geq 1$. Moreover, since $x_{l-1}(a,b) \neq x^{\star}$, we have $f_l(a,b_1,\ldots,b_{l-1}) = r_{j_l}$. Again, we can distinguish two sub-cases:
\begin{itemize}
\item {\bf Case $f_l(a,b_{1},\ldots,b_{l-1}) = b_l$:} In this case, the $d$-process decreases one step. However, by construction $r_{j_l}$ is one of the `successful' flips for the Sch\"oning-process, hence the $x$-process also decreases one step.  By assumption the inequality  (\ref{adfbdfn}) is satisfied for  $l-1,a,b$, and since both the $x$-process and the $d$-process decrease one step, (\ref{adfbdfn}) remains satisfied for $l,a,b$.
\item {\bf Case $f_l(a,b_{1},\ldots,b_{l-1})\neq b_l$:}
In this case, the $d$-process increases one step. The $x$-process may increase or decrease, but with at most one step, so (\ref{adfbdfn}) remains satisfied for $l,a,b$.
\end{itemize}
\end{itemize}
By induction, we can conclude that (\ref{adfbdfn}) is satisfied for all $l,a,b$.
\end{proof}


\subsection{$(\tilde{d}_l)_l$ is a Markov chain}

In the following we wish to show that $(\tilde{d}_l)_l$ is a Markov chain. Recall that  both the genuine Sch\"oning-process  $(x_l)_l$, as well as the walk $(\tilde{d}_l)_l$, are determined by a sequence of  `walk variables' $(B_l)_l$ (and initial-state variable $A$). The  Sch\"oning-walk itself is Markovian, but it is \emph{a priori} not  obvious that the process $(\tilde{d}_l)_l$ also is Markovian, particularly since the $l$-th step of the latter is determined by a complicated function  of all the walk-variables up to the $l$-th step, as described  by  (\ref{fgnsfgnfg}). However, in spite appearances, it turns out that (\ref{fgnsfgnfg}) defines a mapping from the original set of random variables $(B_l)_l$ to a new set of variables  $(\tilde{B}_l)_l$,  in such a manner that  the \emph{change} from $\tilde{d}_{m-1}$ to $\tilde{d}_m$ is determined by $\tilde{B}_m$, and \emph{only} by  $\tilde{B}_m$.  Moreover, it turns out that $(\tilde{B}_l)_l$ is an iid sequence. 
Since all the $\tilde{B}_l$ are independent, it follows that $(\tilde{d}_l)_l$ must be a Markov chain.  In order to show that the new sequence of variables  $(\tilde{B}_l)_l$ is iid, what we actually do is to show that  (\ref{fgnsfgnfg}) induces a bijection on $\{0,1\}^{\times n}\times\{0,1,2\}^{l}$. Since $\big(A, (B_l)_l\big)$ is uniformly distributed (see (\ref{dfabadfbadfn})) it follows by the bijection that  $\big(A, (\tilde{B}_l)_l\big)$ also is uniformly distributed, and thus in particular that $(\tilde{B}_l)_l$ is iid.

\subsection{A bijection}

The following lemma introduces functions $f_s$.  Later, in the proof of Proposition \ref{adfbdfsgnsrgnr}, we will let these mappings be  the functions $f_l(a,b_1,\ldots,b_{l-1})$ in (\ref{sfgnsfgngf}). Since the latter  are algorithmically defined, via the  Sch\"oning-process $(x_l)_l$, it is challenging to get a hold on the properties of these mappings. It is thus worth noting that (apart from domains and ranges) Lemma  \ref{xfgrqreqghr} (and Lemma \ref{mnuzmueeumt}) makes no assumptions on the properties of the mappings $f_s$. Hence, our lack of control over the mappings $f_l(a,b_1,\ldots,b_{l-1})$ will not be an issue in the subsequent proofs.

As preparation, we make the following observations.
\begin{lemma} 
If $t,t',r\in\{0,1,2\}$, then 
\begin{equation}
\label{sfgfsgm}
(t-r)\mathrm{mod}\, 3 = (t'-r)\mathrm{mod}\, 3\quad\Leftrightarrow\quad t = t'.
\end{equation}
Moreover, if $t,r\in\{0,1,2\}$, then
\begin{equation}
\label{sbdfbsf}
\big((t+r)\mathrm{mod}\,3 - r\big)\mathrm{mod}\, 3 = t.
\end{equation}
\end{lemma}

\begin{lemma}
\label{xfgrqreqghr}
Let $f_1:\{0,1\}^{\times n}\rightarrow \{0,1,2\}$ and 
 $f_s:\{0,1\}^{\times n}\times \{0,1,2\}^{\times (s-1)}\rightarrow \{0,1,2\}$ for $s = 2,\ldots, l$ be given. 
Define the  mapping $\{0,1\}^{\times n}\times \{0,1,2\}^{\times l}\ni (b_1,\ldots, b_l) \mapsto Q(a,b_1,\ldots, b_l) = (\tilde{a},\tilde{b}_1,\ldots, \tilde{b}_l)\in \{0,1\}^{\times n}\times\{0,1,2\}^{\times l}$ by
\begin{equation}
\label{sdvabaef}
\begin{split}
\tilde{a} := & a,\\
\tilde{b}_1 := &  \big(b_1-f_1(a)\big)\mathrm{mod}\, 3,\\
\tilde{b}_2 := & \big(b_2-f_2(a,b_1)\big)\mathrm{mod}\, 3,\\
\tilde{b}_3 := & \big( b_3-f_3(a,b_1,b_2)\big)\mathrm{mod}\, 3,\\
\vdots & \\
\tilde{b}_l := & \big( b_l-f_l(a,b_1,\ldots,b_{l-1})\big)\mathrm{mod}\, 3.
\end{split}
\end{equation}
Then $Q$ is a  bijection on $\{0,1\}^{\times n}\times\{0,1,2\}^{\times l}$.
\end{lemma}

\begin{proof}
To show that $Q$ is a bijection, we first show that it is injective, and then that it is surjective.

Let $(a,b_1,\ldots,b_l),(a',b'_1,\ldots,b'_l)\in \{0,1\}^{\times n}\times\{0,1,2\}^{\times l}$ be such that 
\begin{equation}
\label{dvafbadfb}
Q(a,b_1,\ldots,b_l)= Q(a',b'_1,\ldots,b'_l).
\end{equation}
By the first line of (\ref{sdvabaef}) it follows that 
\begin{equation}
\label{dfndgn}
a= \tilde{a} = a'.
\end{equation}
By the second line of (\ref{sdvabaef}) it follows that
\begin{equation}
  \big(b_1-f_1(a)\big)\mathrm{mod}\, 3 = \big(b'_1-f_1(a')\big)\mathrm{mod}\, 3,
\end{equation}
which combined with (\ref{dfndgn}) yields
\begin{equation}
  \big(b_1-f_1(a)\big)\mathrm{mod}\, 3 = \big(b'_1-f_1(a)\big)\mathrm{mod}\, 3.
\end{equation}
Since $f_1(a),b_1,b'_1\in\{0,1,2\}$ it follows by (\ref{sfgfsgm}) that  
\begin{equation}
b_1 = b'_1.
\end{equation}

As an induction hypothesis, assume that  for some $s\geq 2$, it is the case that 
\begin{equation}
a = a',\quad b_{j} = b'_{j},\quad j = 1,\ldots, s-1.
\end{equation}
The $s$th line of (\ref{dvafbadfb}) implies
\begin{equation}
\label{adbvsdfb}
\begin{split}
 & \big(b_s - f_s(a,b_1,\ldots, b_{s-1}) \big)\mathrm{mod}\, 3= \big(b'_s- f_s(a',b'_1,\ldots,b'_{s-1})\big)\mathrm{mod}\, 3.
\end{split}
\end{equation}
By the induction hypothesis, this implies
\begin{equation}
\begin{split}
 & \big(b_s - f_s(a,b_1,\ldots, b_{s-1}) \big)\mathrm{mod}\, 3= \big(b'_s- f_s(a,b_1,\ldots,b_{s-1})\big)\mathrm{mod}\, 3.
\end{split}
\end{equation}
Since $b_s,b'_s,f_s(a,b_1,\ldots,b_{s-1}) \in\{0,1,2\}$, it follows by (\ref{sfgfsgm}) that  
\begin{equation}
b_s = b'_s.
\end{equation}
Since the induction hypothesis is true for $s= 2$, we can conclude that it is true for all $s = 2,\ldots, l$. We can thus conclude that $Q$ is injective.

Next we wish to show that $Q$ is surjective onto  $\{0,1\}^{\times n}\times \{0,1,2\}^{\times l}$.
Let $ (\tilde{a}',\tilde{b}'_1,\ldots,\tilde{b}'_l)\in \{0,1\}^{\times n}\times \{0,1,2\}^{\times l}$. 
Define
\begin{equation}
\label{dfbdgngssng}
\begin{split}
a :=& \tilde{a}',\\
b_1 := & \big(\tilde{b}'_1 +f_1(\tilde{a}')\big)\mathrm{mod}\, 3 = \big(\tilde{b}'_1 +f_1(\tilde{a}')\big)\mathrm{mod}\, 3,
\end{split}
\end{equation}
and the sequence $(b_j)_{j=2}^{l}$ recursively by
\begin{equation}
\label{jutmutuz}
b_{j} := \big(\tilde{b}'_j + f_j(a,b_{j-1},\ldots, b_1)\big)\mathrm{mod}\, 3,\quad j = 2,\ldots, l,
\end{equation}
for $a'$ and $b_1'$ as defined in (\ref{dfbdgngssng}).
In the following, we wish to show that  $Q(a,b_1,\ldots,b_l) =  (\tilde{a}',\tilde{b}'_1,\ldots,\tilde{b}'_l)$. For notational convenience, we introduce the components $Q_0(a,b_1,\ldots,b_l)  := \tilde{a}$ and $Q_j(a,b_1,\ldots,b_l) := \tilde{b}_j$, quad $j =m2,\ldots,l$.

By the first line of (\ref{sdvabaef}) we have
\begin{equation}
\begin{split}
Q_0(a,b_1,\ldots,b_l) = & a = \tilde{a}'.
\end{split}
\end{equation}
By the second line of (\ref{sdvabaef})
\begin{equation}
\begin{split}
Q_1(a,b_1,\ldots,b_l) = & \big(b_1-f_1(a)\big)\mathrm{mod}\, 3\\
= & \Big( \big(\tilde{b}'_1 +f_1(\tilde{a}')\big)\mathrm{mod}\, 3  -f_1(a)\Big)\mathrm{mod}\, 3   \\  
= & \Big( \big(\tilde{b}'_1 +f_1(\tilde{a}')\big)\mathrm{mod}\, 3  -f_1(\tilde{a}')\Big)\mathrm{mod}\, 3   \\  
& [\textrm{By (\ref{sbdfbsf})}]\\
= &  \tilde{b}'_1. 
\end{split}
\end{equation}
For all $j\geq 2$ we moreover have
\begin{equation}
\begin{split}
Q_j(a,b_1,\ldots,b_l) = & \big( b_j-f_j(a,b_1,\ldots,b_{j-1})\big)\mathrm{mod}\, 3\\
& [\textrm{By (\ref{jutmutuz})}]\\
= &  \Big(  \big(\tilde{b}'_j + f_j(a,b_{j-1},\ldots, b_1)\big)   -f_j(a,b_1,\ldots,b_{j-1})\Big)\mathrm{mod}\, 3\\
& [\textrm{By (\ref{sbdfbsf})}]\\
= & \tilde{b}'_j.
\end{split}
\end{equation}
Hence, we can conclude that $Q(a,b_1,\ldots,b_l) =  (\tilde{a}',\tilde{b}'_1,\ldots,\tilde{b}'_l)$. Hence $Q$ is surjective, and thus bijective.
\end{proof}

\subsection{\label{adfbafnsngs}  Transformations that preserve uniformity}

We make the following basic observation
\begin{lemma}
\label{fnblalnnlba}
Let $\mathcal{S}$ be some finite set. Let $Q:\mathcal{S}\rightarrow\mathcal{S}$ be invertible. Let $R$ be some random variable on $\mathcal{S}$. If $R$ is uniformly distributed over $\mathcal{S}$, the $Q(R)$ is also uniformly distributed over $\mathcal{S}$.
\end{lemma}
\begin{proof}
\begin{equation}
\begin{split}
P\big(Q(R) = s\big) = & P\big(R = Q^{-1}(s)\big) = \frac{1}{|\mathcal{S}|}.
\end{split}
\end{equation}
\end{proof}

\begin{lemma}
\label{mnuzmueeumt} 
Let $f_1:\{0,1\}^{\times n}\rightarrow \{0,1,2\}$ and 
 $f_s:\{0,1\}^{\times n}\times \{0,1,2\}^{\times (s-1)}\rightarrow \{0,1,2\}$ for $s = 2,\ldots, l$ be given. 
Assume that $B_1,\ldots,B_l$ are random variables that take values in $\{0,1,2\}$, $A$ be a random variable that takes values in $\{0,1\}^{\times n}$, and that these are distributed as 
\begin{equation}
\label{fdbdaha}
P(A = a, B_1 = b_1,\ldots, B_l = b_l) = \frac{1}{2^n3^l},\quad\forall a\in \{0,1\}^{\times n},\quad \forall (b_1,\ldots, b_l)\in\{0,1,2\}^{\times l}.
\end{equation}
Define $\tilde{B}_1,\ldots, \tilde{B}_l$ by
\begin{equation}
\label{fgnsfgfgnngf}
\begin{split}
\tilde{B}_1 := &  \big(B_1-f_1(A)\big)\mathrm{mod}\, 3,\\
\tilde{B}_2 := & \big(B_2 - f_2(A, B_1)\big)\mathrm{mod}\, 3,\\
\tilde{B}_3 := & \big(B_3  - f_3(A, B_1,B_2)\big)\mathrm{mod}\, 3,\\
\vdots & \\
\tilde{B}_l := & \big(B_l - f_l(A,B_1,\ldots, B_{l-1})\big)\mathrm{mod}\, 3.
\end{split}
\end{equation}
Then 
\begin{equation}
\label{fgmhmdghm}
P(A = a, \tilde{B}_1 = \tilde{b}_1,\ldots, \tilde{B}_l = \tilde{b}_l) = \frac{1}{2^n3^l},\quad\forall a\in \{0,1\}^{\times n},\quad \forall (\tilde{b}_1,\ldots, \tilde{b}_l)\in\{0,1,2\}^{\times l}.
\end{equation} 
Consequently, $A,\tilde{B}_1,\ldots,\tilde{B}_l$ are independent and uniformly distributed.
\end{lemma}

\begin{proof}
By (\ref{fdbdaha}) we know that $(A,B_1,\ldots,B_l)$ is uniformly distributed on $\{0,1\}^{\times n}\times \{0,1,2\}^{\times l}$.
With the additional definition $\tilde{A} := A$, we note that (\ref{fgnsfgfgnngf}) can be rewritten as
\begin{equation}
(\tilde{A},\tilde{B}_1,\ldots,\tilde{B}_l) := Q(A,B_1,\ldots,B_l),
\end{equation}
where $Q:\{0,1\}^{\times n}\times \{0,1,2\}^{\times l}\rightarrow \{0,1\}^{\times n}\times \{0,1,2\}^{\times l}$ is as defined in Lemma
\ref{xfgrqreqghr}. By Lemma \ref{xfgrqreqghr}, we moreover know that $Q$ is a bijection on $\{0,1\}^{\times n}\times \{0,1,2\}^{\times l}$ and thus invertible. Hence, by Lemma \ref{fgmhmdghm}, we know that $(\tilde{A},\tilde{B}_1,\ldots,\tilde{B}_l)$ also is uniformly distributed over $\{0,1\}^{\times n}\times \{0,1,2\}^{\times l}$. Since $\tilde{A} = A$, we can conclude that (\ref{fgmhmdghm}) holds. 

By (\ref{fgmhmdghm}) i follows that 
\begin{equation}
\begin{split}
& P(A = a, \tilde{B}_1 = \tilde{b}_1,\ldots, \tilde{B}_l = \tilde{b}_l) = P(A = a)P(\tilde{B}_1 = \tilde{b}_1)\cdots P(\tilde{B}_l = \tilde{b}_l),\\
& P(A =a) = \frac{1}{2^n},\quad P(\tilde{B}_1 = \tilde{b}_1) = \frac{1}{3},\ldots,P(\tilde{B}_l = \tilde{b}_l) = \frac{1}{3}.
\end{split}
\end{equation}
and thus $A$, $\tilde{B}_1,\ldots,\tilde{B}_l$ are independent and uniformly distributed.

\end{proof}


\subsection{\label{ndtehntdhmt} The process $(\tilde{d}_l)_{l}$ is a Markov chain}

\begin{proposition}
\label{adfbdfsgnsrgnr}
Let $(\tilde{d}_l)_{l}$ be the process as defined by (\ref{sfgnsfgngf}) and (\ref{fgnsfgnfg}), with respect  to the variables $A$, $B_1,B_2,\ldots$ distributed as in (\ref{dfabadfbadfn}). For each $m$ there exist
variables $\tilde{B}_1,\cdots,\tilde{B}_m$ that are iid and uniformly distributed on $\{0,1,2\}$, and are independent of $A$, such  that
\begin{equation}
\label{fdnsfnmAgain}
\begin{split}
\tilde{d}_0 := &  d_H(A,x^{\star}),\\
\tilde{d}_{l} := & \left\{\begin{matrix}
0 & \mathrm{if} & \tilde{d}_{l-1} = 0\\
\tilde{d}_{l-1} + 1 & \mathrm{if} & \tilde{d}_{l-1} \neq 0,\quad  \tilde{B}_l  \neq 0\\
\tilde{d}_{l-1} - 1 & \mathrm{if} &  \tilde{d}_{l-1} \neq 0,\quad
\tilde{B}_l  = 0
\end{matrix}\right. \quad l = 1,2,\ldots,m
\end{split}
\end{equation}
Hence, $(\tilde{d}_l)_{l}$ is a Markov chain described by the transition probabilities
\begin{equation}
\label{fgmhmdghdg}
P(\tilde{d}_{l+1} = j|\tilde{d}_{l} = k) = \delta_{j,0}\delta_{k,0} + (1-\delta_{k,0})\big(\frac{1}{3}\delta_{j,k-1} +\frac{2}{3}\delta_{j,k+1}\big),\quad \forall j,k\in\mathbb{N},\quad \forall l
\end{equation}
with initial distribution
\begin{equation}
P(\tilde{d}_0 = j) = P\big(d_H(A,x^{\star})= j\big).
\end{equation}
Moreover, for the distribution of $A$ as in  (\ref{dfabadfbadfn}) we have
\begin{equation}
\label{afadaanen}
P(\tilde{d}_0 = j) = \left\{\begin{matrix} \frac{1}{2^n}\binom{n}{j}&,\quad 0\leq j\leq n\\
0&  \mathrm{otherwise}
\end{matrix}\right.
\end{equation}

\end{proposition}
In (\ref{fgmhmdghdg}), the term $\delta_{j,0}\delta_{k,0}$ signifies $d = 0$ being an absorbing state. In the second term, the effect of the factor  $(1-\delta_{k,0})$ is that if the chain is not in the absorbing state, then the transition probabilities are given by $\frac{1}{3}\delta_{j,k-1} +\frac{2}{3}\delta_{j,k+1}$. Hence, with probability $1/3$, it takes a step `down', and with probability $2/3$ it takes a step `up'.

\begin{proof}

For $t,r\in \{0,1,2\}$ it is the case that 
\begin{equation}
t = r \quad\Leftrightarrow \quad (t-r)\mathrm{mod}\,3 = 0.
\end{equation}
By this observation, it follows that (\ref{fgnsfgnfg}) can be rewritten as
\begin{equation}
\label{badntrnwrt}
\begin{split}
\tilde{d}_0(a,b) := &  d_H(a,x^{\star}),\\
\tilde{d}_{l}(a,b_{1},\ldots, b_{l}) := & \left\{\begin{matrix}
0 & \mathrm{if} & \tilde{d}_{l-1}(a,b_{1},\ldots,b_{l-1}) = 0\\
\tilde{d}_{l-1}(a,b_{1},\ldots, b_{l-1}) + 1 & \mathrm{if} & \tilde{d}_{l-1}(a,b_{1},\ldots, b_{l-1}) \neq 0,\quad  \big(b_l -f_l(a,b_{1},\ldots, b_{l-1})\big)\mathrm{mod}\,3  \neq 0\\
\tilde{d}_{l-1}(a,b_{1},\ldots, b_{l-1}) - 1 & \mathrm{if} &  \tilde{d}_{l-1}(a,b_{1},\ldots, b_{l-1}) \neq 0,\quad
 \big(b_l -f_l(a,b_{1},\ldots, b_{l-1})\big)\mathrm{mod}\,3  = 0
\end{matrix}\right. \quad l = 1,2,\ldots
\end{split}
\end{equation}
Next, we rewrite (\ref{badntrnwrt}) such that we suppress the explicit dependence on the elementary events $(a,b)$.
\begin{equation}
\label{fgnsfgmsmfg}
\begin{split}
\tilde{d}_0 := &  d_H(A,x^{\star}),\\
\tilde{d}_{l} := & \left\{\begin{matrix}
0 & \mathrm{if} & \tilde{d}_{l-1} = 0\\
\tilde{d}_{l-1} + 1 & \mathrm{if} & \tilde{d}_{l-1} \neq 0,\quad  \big(B_l -f_l(A,B_{1},\ldots, B_{l-1})\big)\mathrm{mod}\,3  \neq 0\\
\tilde{d}_{l-1} - 1 & \mathrm{if} &  \tilde{d}_{l-1} \neq 0,\quad
 \big(B_l -f_l(A,B_{1},\ldots, B_{l-1})\big)\mathrm{mod}\,3  = 0
\end{matrix}\right. \quad l = 1,2,\ldots
\end{split}
\end{equation}
If we define $\tilde{B}_1,\ldots, \tilde{B}_l$ by
\begin{equation}
\label{fgnsfgfgnngf3}
\begin{split}
\tilde{B}_1 := &  f_1(A),\\
\tilde{B}_2 := & \big(B_2 - f_2(A, B_1)\big)\mathrm{mod}\, 3,\\
\tilde{B}_3 := & \big(B_3  - f_3(A, B_1, B_2)\big)\mathrm{mod}\, 3,\\
\vdots & \\
\tilde{B}_l := & \big(B_l - f_l(A,B_1,\ldots, B_{l-1})\big)\mathrm{mod}\, 3,
\end{split}
\end{equation}
then we can rewrite (\ref{fgnsfgmsmfg}) as
\begin{equation}
\label{fdnsfnm}
\begin{split}
\tilde{d}_0 := &  d_H(A,x^{\star}),\\
\tilde{d}_{l} := & \left\{\begin{matrix}
0 & \mathrm{if} & \tilde{d}_{l-1} = 0\\
\tilde{d}_{l-1} + 1 & \mathrm{if} & \tilde{d}_{l-1} \neq 0,\quad  \tilde{B}_l  \neq 0\\
\tilde{d}_{l-1} - 1 & \mathrm{if} &  \tilde{d}_{l-1} \neq 0,\quad
\tilde{B}_l  = 0
\end{matrix}\right. \quad l = 1,2,\ldots
\end{split}
\end{equation}
By Lemma \ref{mnuzmueeumt} we know that $A,\tilde{B}_1,\ldots,\tilde{B}_l$ are independent and uniformly distributed.
Since the l:th step is determined solely by $\tilde{B}_l$, and these are independent of each other, and of $A$, it follows that $(\tilde{d}_l)_{l}$ is a Markov chain. 
By inspecting (\ref{fdnsfnm}) we first see that 
\begin{equation}
\label{dfngng}
\begin{split}
P(\tilde{d}_l = j|\tilde{d}_{l-1} = 0)  = \delta_{j,0},
\end{split}
\end{equation}
while for $\tilde{d}_{l-1} = k\neq 0$ we have
\begin{equation}
\label{gfnsfgngf}
\begin{split}
P(\tilde{d}_l = j|\tilde{d}_{l-1} = k)  = & \delta_{j,k+1} P(\tilde{B}_l  \neq 0) + \delta_{j,k-1}P(\tilde{B}_l  = 0)\\
= & \frac{2}{3}\delta_{j,k+1} + \frac{1}{3}\delta_{j,k-1},
\end{split}
\end{equation}
where the last step follows since each $\tilde{B}_l$ is uniformly distributed over $\{0,1,2\}$.
By combining the cases (\ref{dfngng}) and (\ref{gfnsfgngf}) we obtain (\ref{fgmhmdghdg}). By (\ref{fdnsfnm}) it moreover follows that $P(\tilde{d}_0 = j) = P\big(d_H(A,x^{\star})= j\big)$. Since $A$ is uniformly distributed over $\{0,1\}^{\times n}$, it means that $d_H(A,x^{\star})$ is binomially distributed. Thus for $0\leq j \leq n$, we have $P(\tilde{d}_0 = j) =\frac{1}{2^n}\binom{n}{j}$.

\end{proof}

\subsection{Relating probabilities of $(x_l)_l$ and $(\tilde{d}_l)_l$}

The reason for why we introduce the walk $(\tilde{d}_l)_l$ is in order to bound the relevant success-probabilities of the more complicated true Sch\"oning-walk $(x_l)_l$.  The lemma below considers two such inequalities, which we will use when we determine the bounds for the Groverized walk. 

 \begin{lemma}
\label{TransitionFromxTod}
Let $(x_l)_{l\in\mathbb{N}}$ be the Sch\"oning process for bit strings of length $n$, with $x^{\star}$ the selected satisfying assignment. Let $(\tilde{d}_l)_{l}$ be the process as defined by (\ref{sfgnsfgngf}) and (\ref{fgnsfgnfg}).  Then
\begin{equation}
\label{wrtnwrntwrt}
\begin{split}
 P(x_m  = x^{\star})\geq P(\tilde{d}_m = 0),
\end{split}
\end{equation}
\begin{equation}
\label{dfbwbwrttrn}
\begin{split}
&  P\big(  x_m = x^{\star}|d_H(x_0,x^{\star}) = j\big)\geq P(\tilde{d}_m = 0|\tilde{d}_0 = j), 
\end{split}
\end{equation}

\end{lemma}

\begin{proof}
We begin by proving inequality (\ref{wrtnwrntwrt}). For the sake of notational simplicity, we let $\omega$ denote the elements of the event space (where we could regard $\omega$ as $(a,b)$ or $(a,\tilde{b})$).
By Lemma \ref{bound} we know that 
\begin{equation}
\begin{split} 
& d_H\big(x_m(\omega),x^{\star}\big) \leq \tilde{d}_m(\omega),
\end{split}
\end{equation}
which  implies
\begin{equation}
\begin{split}
& \tilde{d}_m(\omega) = 0\quad \Rightarrow \quad d_H\big(x_m(\omega),x^{\star}\big) =0\\
\end{split}
\end{equation}
and thus
\begin{equation}
\label{fdnsnnnrnwm}
\begin{split}
& \{\omega: \tilde{d}_m(\omega) = 0\} \subset \{\omega: d_H\big(x_m(\omega),x^{\star}\big) = 0\} = \{\omega: x_m(\omega) = x^{\star}\}
\end{split}
\end{equation}
and thus 
\begin{equation}
\begin{split}
P(\tilde{d}_m = 0) & = P\big(\{\omega: \tilde{d}_m(\omega) = 0\}\big) \\
& \leq P\big(  \{\omega: x_m(\omega) = x^{\star}\}\big)\\
& = P(x_m  = x^{\star}),
\end{split}
\end{equation}
which proves (\ref{wrtnwrntwrt}).


We next turn to the proof of (\ref{dfbwbwrttrn})
By definition of the walk $(\tilde{d}_l)_l$ we have $\tilde{d}_0(\omega) = d_H\big(x_0(\omega),x^{\star}\big)$, and thus 
\begin{equation}
\label{muzfmzurz}
\{\omega: \tilde{d}_0(\omega) = j\} = \{\omega: d_H\big(x_0(\omega),x^{\star}\big) = j\}
\end{equation}
and consequently
\begin{equation}
\label{afbaefbebnen}
P(\tilde{d}_0 = j) = P\big(d_H(x_0,x^{\star}) = j\big).
\end{equation}
By combining (\ref{fdnsnnnrnwm})  and (\ref{muzfmzurz}) we obtain
\begin{equation}
\{\omega: \tilde{d}_m(\omega) = 0\}\cap \{\omega: \tilde{d}_0(\omega) = j\} \subset \{\omega: x_m(\omega) = x^{\star}\}\cap\{\omega: d_H\big(x_0(\omega),x^{\star}\big) = j\}
\end{equation}
and consequently
\begin{equation}
\begin{split}
P(\tilde{d}_m = 0,\tilde{d}_0 = j) 
= & P\big(\{\omega: \tilde{d}_m(\omega) = 0\}\cap \{\omega: \tilde{d}_0(\omega) = j\}\big)\\
\leq & P\big( \{\omega: x_m(\omega) = x^{\star}\}\cap\{\omega: d_H\big(x_0(\omega),x^{\star}\big) = j\}\big)\\
= &  P\big(  x_m = x^{\star}, d_H(x_0,x^{\star}) = j\big).
\end{split}
\end{equation}
By combining this with (\ref{afbaefbebnen}) we can conclude that
\begin{equation}
\begin{split}
P(\tilde{d}_m = 0|\tilde{d}_0 = j) 
\leq
&  P\big(  x_m = x^{\star}|d_H(x_0,x^{\star}) = j\big),
\end{split}
\end{equation}
which proves (\ref{dfbwbwrttrn}).


\end{proof}

\subsection{From walks  on $\mathbb{N}$ to  walks  on $\mathbb{Z}$}
So far, we have replaced the projection of the Sch\"oning-process $(x_l)_l$ to the Hamming distance $d_H(x_m,x^{\star})$  with the substitute Markov-chain $(\tilde{d}_l)_l$. Similar to $x^{\star}$ being an absorbing state of $(x_l)$, the process $(\tilde{d}_l)_l$ has $0$ as absorbing state. As a model of the true Sch\"oning process, this  absorbing state certainly makes sense, since it  corresponds to a setting where we at each step monitor whether a solution has been reached, and the process is terminated once this happens.  For  the sake of obtaining tractable expressions for the relevant probabilities, we  here take one step further and instead consider a walk on $\mathbb{Z}$.  Analogously to how Lemma \ref{TransitionFromxTod} bounds the relevant probabilities of the true Sch\"oning process, with the corresponding quantities in $(\tilde{d}_l)_l$, Lemma \ref{TransitionFromdToTilded} below, bounds the relevant probabilities of $(\tilde{d}_l)_l$ in terms of corresponding quantities for a Markov-chain $(d_l)_l$ extended to the whole of $\mathbb{Z}$.

As a bit of a side remark, one may note that the results in (\ref{TransitionFromdToTilded}) does not necessarily refer to the particular Markov-chain defined by  (\ref{sfgnsfgngf}) and (\ref{fgnsfgnfg}), but could be any Markov chain on $\mathbb{N}$ with fixed transition probabilities and absorbing boundary condition at $0$.

\begin{lemma}
\label{TransitionFromdToTilded}
Let $(\tilde{d}_l)_{l\in\mathbb{N}}$ be a Markov chain on $\mathbb{N}$, with transition probabilities
\begin{equation}
\label{dbangnan}
P(\tilde{d}_{l+1} = j|\tilde{d}_{l} = k) = \delta_{j,0}\delta_{k,0} + (1-\delta_{k,0})\big((1-q)\delta_{j,k-1} +q\delta_{j,k+1}\big),\quad \forall j,k\in\mathbb{N},\quad \forall l\in\mathbb{N},
\end{equation}
for some $0\leq q\leq 1$. 
Let $(d_l)_{l\in\mathbb{N}}$ be a Markov chain on $\mathbb{Z}$, with transition probabilities
\begin{equation}
\label{gfndnnezm}
P(d_{l+1} = j|d_{l} = k) = (1-q)\delta_{j,k-1} + q\delta_{j,k+1},\quad \forall j,k\in\mathbb{Z},\quad \forall. l\in\mathbb{N}.
\end{equation}
Then
\begin{equation}
\label{adfnrtnrtnjwrtj}
P\big(d_m\leq 0|d_0 = j\big) \leq P(\tilde{d}_m = 0|\tilde{d}_0 = j),\quad \forall m\in\mathbb{N},\quad \forall j\in \mathbb{N}.
\end{equation}
Consequently, if the initial state $d_0$ is such that 
\begin{equation}
\label{shtmtsmtemtz}
P(d_0 = j) = \left\{\begin{matrix}
P(\tilde{d}_0 = j), & \quad j \geq 0\\
0 , & \quad  j <0
\end{matrix}\right.
\end{equation}
then
\begin{equation}
\label{sfgndmdhmhm}
P(d_m = 0) \leq P(d_m\leq 0) \leq P(\tilde{d}_m = 0),\quad \forall m\in\mathbb{N}.
\end{equation}
\end{lemma}

\begin{proof}

For notational convenience, we define
\begin{equation}
\label{netntmeh}
\begin{split}
M_{j,k} := &  P(\tilde{d}_{l+1} = j|\tilde{d}_{l} = k),
\end{split}
\end{equation}
and
\begin{equation}
\label{sfgsfgnsfgnswnr}
\begin{split}
\tilde{M}_{j,k} := & P(d_{l+1} = j|d_{l} = k).
\end{split}
\end{equation}
By comparing with (\ref{dbangnan}) and (\ref{gfndnnezm}) one can see that
\begin{equation}
\label{adfbsfgnf}
\begin{split}
M_{j,k} = \tilde{M}_{j,k},\quad \forall j>0,\quad \forall k>0.
\end{split}
\end{equation}
We note that $0$ is an absorbing state for $(\tilde{d}_l)_l$. Hence, 
\begin{equation}
\begin{split}
 \tilde{d}_{s-1} = 0\quad\Rightarrow\quad \tilde{d}_{s} = 0,
\end{split}
\end{equation}
which implies
\begin{equation}
\begin{split}
\tilde{d}_s>0 \quad \Rightarrow\quad \tilde{d}_{s-1}>0,
\end{split}
\end{equation}
which in turn implies
\begin{equation}
\label{sbadngnngr}
\begin{split}
 P(\tilde{d}_s= k_s|\tilde{d}_{s-1}=0) = 0,\quad\textrm{if}\quad k_s >0.
\end{split}
\end{equation}

We begin by proving (\ref{adfnrtnrtnjwrtj}). For this purpose, assume that $j>0$. 
\begin{equation}
\begin{split}
 &P(\tilde{d}_l >0 |\tilde{d}_0 = j)\\
 = & \sum_{k_l>0}P(\tilde{d}_l =k_l |\tilde{d}_0 = j)\\
& [\textrm{By Markovianity}]\\
= & \sum_{k_l>0}\sum_{k_{l-1},\ldots,k_1} P(\tilde{d}_l = k_l|\tilde{d}_{l-1} = k_{l-1})P(\tilde{d}_{l-1} = k_{l-1}|\tilde{d}_{l-2} = k_{l-2})\cdots P(\tilde{d}_2 = k_2|\tilde{d}_1 = k_1)P(\tilde{d}_1 = k_1|\tilde{d}_0 = j)\\
= & \sum_{k_l>0}\sum_{k_{l-1}:k_{l-1}>0}\sum_{k_{l-2},\ldots,k_1} P(\tilde{d}_l = k_l |\tilde{d}_{l-1} = k_{l-1})P(\tilde{d}_{l-1} = k_{l-1}|\tilde{d}_{l-2} = k_{l-2})\cdots  P(\tilde{d}_1 = k_1|\tilde{d}_0 = j)\\
& + \sum_{k_l>0}\sum_{k_{l-2},\ldots,k_1} P(\tilde{d}_l = k_l|\tilde{d}_{l-1} = 0)P(\tilde{d}_{l-1} = 0|\tilde{d}_{l-2} = k_{l-2})\cdots P(\tilde{d}_1 = k_1|\tilde{d}_0 = j)\\
& [\textrm{Since $k_l> 0$ it follows by (\ref{sbadngnngr}) that $P(\tilde{d}_l = k_l|\tilde{d}_{l-1} = 0) = 0$.}]\\
= & \sum_{k_l>0}\sum_{k_{l-1}:k_{l-1}>0}\sum_{k_{l-2},\ldots,k_1} P(\tilde{d}_l = k_l |\tilde{d}_{l-1} = k_{l-1})P(\tilde{d}_{l-1} = k_{l-1}|\tilde{d}_{l-2} = k_{l-2})\cdots P(\tilde{d}_1 = k_1|\tilde{d}_0 = j)\\
& [\textrm{By iteration}]\\
= & \sum_{k_l>0}\sum_{k_{l-1},\ldots,k_1:k_{l-1} >0,\ldots, k_1 >0} P(\tilde{d}_l = k_l |\tilde{d}_{l-1} = k_{l-1})P(\tilde{d}_{l-1} = k_{l-1}|\tilde{d}_{l-2} = k_{l-2})\cdots  P(\tilde{d}_1 = k_1|\tilde{d}_0 = j)\\
& [\textrm{By (\ref{netntmeh})}]\\
= & \sum_{k_l>0} \sum_{k_{l-1},\ldots,k_1:k_{l-1} >0,\ldots, k_1 >0} M_{k_l,k_{l-1}}\cdots M_{k_1,j}\\
& [\textrm{Since $k_l > 0,\, k_{l-1}>0,\cdots, k_{1}>0, j>0$ if follows by (\ref{adfbsfgnf}) that}] \\
= &  \sum_{k_l>0} \sum_{k_{l-1},\ldots,k_1:k_{l-1} >0,\ldots, k_1 >0} \tilde{M}_{k_l,k_{l-1}}\cdots \tilde{M}_{k_1,j}\\
& [\quad \tilde{M}_{k_l,k_{l-1}}\geq 0  \quad]\\
\leq  &  \sum_{k_l>0}  \sum_{k_{l-1},\ldots, k_{1}}\tilde{M}_{k_l,k_{l-1}}\cdots \tilde{M}_{k_1,j}\\
& [\textrm{By (\ref{sfgsfgnsfgnswnr})}]\\
= &   \sum_{k_l>0} \sum_{k_{l-1},\ldots, k_{1}}P(d_l = k_l |d_{l-1} = k_{l-1})P(d_{l-1} = k_{l-1}|d_{l-2} = k_{l-2})\cdots  P(d_1 = k_1|d_0 = j)\\
& [\textrm{By the Markovianity}]\\
= &   \sum_{k_l>0}P(d_l = k_l|d_0 =j)\\
= & P(d_l > 0|d_0 =j).
\end{split}
\end{equation}
Consequently
\begin{equation}
\begin{split}
P(\tilde{d}_l =0 |\tilde{d}_0 = j)=  &1- P(\tilde{d}_l >0 |\tilde{d}_0 = j)\\
\geq   & 1-P(d_l > 0|d_0 =j)\\
= & P(d_l \leq 0|d_0 =j),\quad j>0.
\end{split}
\end{equation}
In the case $\tilde{d}_0 = 0$ we know that this is an absorbing state, and thus $P(\tilde{d}_l =0 |\tilde{d}_0 = 0) = 1$. Consequently, $P(d_l\leq 0|d_0 = 0)\leq P(\tilde{d}_l =0 |\tilde{d}_0 = 0) = 1$.
This thus proves the inequality in  (\ref{adfnrtnrtnjwrtj}).

With the initial distribution (\ref{shtmtsmtemtz}) we find
\begin{equation}
\begin{split}
P(d_l\leq 0) = & \sum_{j\in\mathbb{Z}}P(d_l\leq 0|d_0 = j)P(d_0 = j)\\
= & \sum_{j\geq 0}P(d_l\leq 0|d_0 = j)P(\tilde{d}_0 = j)\\
\leq  & \sum_{j\geq 0}P(\tilde{d}_l = 0|\tilde{d}_0 = j)P(\tilde{d}_0 = j)\\
= & P(\tilde{d}_l = 0).
\end{split}
\end{equation}

\end{proof}


\subsection{Bounds for Sch\"oning Walks and Groverized Walks}

Here we combine the previous observations in order to obtain the following lower bounds on the success-probability of the  Sch\"oning process. We also obtain the inequalites needed for determining the desired bound on the success-probability of the the Groverized walk.
\begin{proposition}
\label{gnsfgnsgm}
Let $(x_l)_{l\in\mathbb{N}}$ be the Sch\"oning process for bit strings of length $n$, with $x^{\star}$ the selected satisfying assignment. Let $(\tilde{d}_l)_{l}$ be the process as defined by (\ref{sfgnsfgngf}) and (\ref{fgnsfgnfg}).
Let $(d_l)_l$ be the Markov chain as defined by the transition probabilites (\ref{gfndnnezm}) for $q = 2/3$  in Lemma \ref{TransitionFromdToTilded}, for the initial state
\begin{equation}
\label{savnosdvn2}
P(d_0 = j) = \left\{\begin{matrix}
P(\tilde{d}_0 = j) =  P\big(d_H(x_0,x^{\star}) =j\big) = \frac{1}{2^n}\binom{n}{j}, & \quad n\geq j \geq 0\\
0 , & \quad  \mathrm{otherwise}
\end{matrix}\right.
\end{equation}
 Then
\begin{equation}
\label{ybfdfnggf}
\begin{split}
P(x_m = x^{\star}) \geq  &  P(\tilde{d}_m= 0) \\
 \geq &  P(d_m\leq 0)\\
 = &   \sum_{\substack{j,l: j+ m-2l \leq 0,\\ 0 \leq j\leq n,\\ 0\leq l\leq m}}  \frac{1}{2^n}\binom{n}{j} \binom{m}{l}  \left(\frac{1}{3}\right)^{l}\left(\frac{2}{3}\right)^{m-l}
\end{split}
\end{equation}
and
\begin{equation}
\label{dvafebrtnrtnw}
\begin{split}
P\big(x_m = x^{\star}\big|d_H(x_0|x^{\star}) = j\big) \geq  & P(\tilde{d}_m =  0|\tilde{d}_0 = j) \\
 \geq & P(d_m\leq 0|d_0 = j)\\
 = & \sum_{\substack{l: j+ m-2l \leq 0,\\ 0\leq l\leq m}}    \binom{m}{l}  \left(\frac{1}{3}\right)^{l}\left(\frac{2}{3}\right)^{m-l}.
\end{split}
\end{equation}
\end{proposition}

\begin{proof}

By (\ref{wrtnwrntwrt}) in Lemma \ref{TransitionFromxTod} yields the  first  inequality in  (\ref{ybfdfnggf}).

By Proposition \ref{adfbdfsgnsrgnr} we know that $(\tilde{d}_l)_l$ is a Markov chain with transition probabilities as in (\ref{fgmhmdghdg}) and initial distribution (\ref{afadaanen}). By these observations, it follows that the  second inequality  in (\ref{ybfdfnggf}) is a direct application of (\ref{sfgndmdhmhm}) in Lemma \ref{TransitionFromdToTilded}. 
By Lemma \ref{TransitionFromdToTilded}, we also know that the Markov chain $(d_l)_l$ defined by the transition probabilities
\begin{equation}
\label{bdfnfgngf}
P(d_{l+1} = j|d_{l} = k) = \frac{1}{3}\delta_{j,k-1} + \frac{2}{3}\delta_{j,k+1},\quad \forall j,k\in\mathbb{Z},\quad \forall l\in\mathbb{N}.
\end{equation}
and initial distribution
\begin{equation}
\label{adfbdfbnf}
P(d_0 = j) = \left\{\begin{matrix}
P(\tilde{d}_0 = j), & \quad j \geq 0\\
0 , & \quad  j <0
\end{matrix}\right. \quad=  \left\{\begin{matrix} \frac{1}{2^n}\binom{n}{j}, & 0\leq j\leq n\\
0, & \mathrm{otherwise}
\end{matrix}\right.
\end{equation}
From (\ref{bdfnfgngf}) it follows that 
\begin{equation}
\label{adfntnwrtnrt}
P(d_m\leq 0|d_0 = j) = \sum_{l: j+m-2l\leq 0,\, 0\leq l\leq m} \binom{m}{l}  \left(\frac{1}{3}\right)^{l}\left(\frac{2}{3}\right)^{m-l}
\end{equation}
and thus 
\begin{equation}
\label{afdbajsrtjsj}
\begin{split}
P(d_m\leq 0) = & \sum_j P(d_m\leq 0|d_0 = j)P(d_0 = j)\\
& [\textrm{By (\ref{adfbdfbnf})}]\\
 = & \sum_{\substack{j,l: j+ m-2l \leq 0,\\ 0 \leq j\leq n,\\ 0\leq l\leq m}}  \frac{1}{2^n}\binom{n}{j} \binom{m}{l}  \left(\frac{1}{3}\right)^{l}\left(\frac{2}{3}\right)^{m-l}.
\end{split}
\end{equation}


Next, we turn to the inequalities in (\ref{dvafebrtnrtnw}). 
Inequality (\ref{dfbwbwrttrn}) in Lemma \ref{TransitionFromxTod} yields the first inequality in (\ref{dvafebrtnrtnw}). The second inequality in (\ref{dvafebrtnrtnw}) is a direct application of 
(\ref{adfnrtnrtnjwrtj}) in Lemma (\ref{TransitionFromdToTilded}). By (\ref{adfntnwrtnrt}) we already know the last equality in (\ref{dvafebrtnrtnw}).

\end{proof}


\subsection{Relation to the leading order analysis of the Sch\"oning and  GW process}
Here we connect to the analysis of the asymptotic scaling in the main text, by obtaining the starting points, so to speak, of the leading order analysis of the Sch\"oning process and the GW process.

\subsubsection{Sch\"oning process}
For the Sch\"oning process, the average number of repetitions needed to find a solution is given by 
\begin{equation}
\label{bfgnngf}
N_{\textrm{Sch\"oning}} = \frac{1}{P(x_m = x^{\star})}.
\end{equation}
 By sequences of lower bounds on the ideal success-probability $P(x_m = x^{\star})$, we thus obtain upper bounds on $N_{\textrm{Sch\"oning}}$. The step from the true Sch\"oning process to the walk on $Z$ corresponds to one such inequality, i.e. to 
 \begin{equation}
\label{nmzuzr}
 P(x_m = x^{\star}) \geq    P(d_m\leq 0)
 \end{equation}
  in \eqref{ybfdfnggf} in Proposition \eqref{gnsfgnsgm}.
The leading-order analysis in the main text is based on further such inequalities, with the rationale that the `loss' of probability weight becomes irrelevant for the rates $\gamma = \lim_{n\rightarrow\infty} \frac{1}{n}\log N_{\textrm{Sch\"oning}}$, if the inequalities are chosen to correspond to the leading order contributions.  As a first step along these lines, we restrict to an event where we not only reach the desired solution, but also start the system $x_0$ at Hamming distance  $d_H(x_0,x^{\star}) = j$.
 Trivially, 
\begin{equation}
\label{fgnngsgf}
 P(d_m\leq 0) \geq P(d_m\leq 0, d_0 = j) = P(d_0 = j)P(d_m \leq 0|d_0 = j),
\end{equation}
where can identify $P(d_0 = j)$ with $P(E_1)$ in the main text, i.e. 
\begin{equation}
\label{gfnfsgmgsfm}
P(d_0 = j)  = P(E_1) = \frac{1}{2^n}\binom{n}{\kappa n},\quad j = \kappa n. 
\end{equation}
Next, we wish to connect the remaining factor  in \eqref{fgnngsgf}, i.e.,  $P(d_m\leq 0|d_0 = j)$, to the probability $P(E_2)$, which we recall from the main text corresponds to the event $E_2$, where precisely $\nu m$ steps decrease the Hamming distance, while precisely $(1-\nu)m$ steps increase the Hamming distance. (For the walk on $\mathbb{Z}$ this extends to  $\nu m$ steps in the negative direction, and $(1-\nu)m$ steps in the positive direction.) We conclude that the total decrease is 
\begin{equation}
\label{fsgnsfgn}
d_0-d_m = (2\nu-1)m.
\end{equation}
Let us also recall that the combination of  $E_1$ and $E_2$ is successful, i.e. leads to $d_m\leq 0$, if 
\begin{equation}
(2\nu -1)m\geq \kappa n.
\end{equation}
It is useful to note that $(d_l)_l$ is not only Markovian, but also translation symmetric, which means that the change $d_0-d_m$ is independent of the initial state $d_0$, i.e., the joint distribution of these factorize. (As a side remark, this independence also means that $P(E_1\cap E_2) = P(E_1)P(E_2)$.)
 Hence,
\begin{equation}
\label{jutzmuzm}
\begin{split}
& P(d_{m}\leq 0|d_0 = \kappa n)\\
& = P(d_0-d_{m}\geq \kappa n|d_0 = \kappa n)\\
& [\textrm{Since $d_0-d_m$ is independent of $d_0$}]\\
& = P(d_0-d_{m}\geq \kappa n).
\end{split}
\end{equation}
By comparison of \eqref{jutzmuzm} with \eqref{fsgnsfgn} it follows that 
\begin{equation}
\label{dfgndghmdhg}
\begin{split}
& P(d_{m}\leq 0|d_0 = \kappa n)\\
& = P(d_0-d_{m}\geq \kappa n)\\
&\geq P(E_2),\quad \mathrm{if}\quad (2\nu -1)m\geq \kappa n. 
\end{split}
\end{equation}
Alternatively, we can reach the same conclusion by comparing \eqref{eqn:events_probs} with  \eqref{adfntnwrtnrt} to see  that  
\begin{equation}
\label{sfgnsfmfm}
P(E_2) =  \binom{m}{\nu m}\left(\frac{1}{3}\right)^{\nu m}\left(\frac{2}{3}\right)^{(1-\nu)m}
\leq P(d_m\leq 0|d_0 = \kappa n),\quad  \mathrm{if}\quad (2\nu -1)m\geq \kappa n
\end{equation}
By \eqref{bfgnngf}, \eqref{nmzuzr}, \eqref{fgnngsgf}, \eqref{gfnfsgmgsfm} and \eqref{dfgndghmdhg}, we can conclude that
\begin{equation}
N_{\textrm{Sch\"oning}}\leq \frac{1}{P(E_1)P(E_2)},\quad \mathrm{if}\quad (2\nu -1)m\geq \kappa n.
\end{equation}

\subsubsection{GW process}

For the GW process, let us recall that it consists of a classical outer loop that at each round assigns a definite (classical) initial state, while the walk-process is Groverized. We assume that the number of iterations of the Grover-procedure is tuned to the density of successful walks, for a specific initial Hamming distance $j = \kappa n$, i.e., to the success-probability $P\big(x_m = x^{\star}\big|d_H(x_0|x^{\star}) = \kappa n\big)$. In the analysis we lower bound the success-probability by assuming that process fails whenever $d_H(x_0,x^{\star}) \neq \kappa n$ (which may be pessimistic).  The probability to obtain the initial state $x_0$ with Hamming distance $\kappa n$ is $P\big(d_H(x_0,x^{\star}) =\kappa n\big)$, and thus in average we need to repeat the outer loop $1/P\big(d_H(x_0,x^{\star}) =\kappa n\big)$ times to be guaranteed to at least once reach the initial Hamming distance $\kappa n$. In the successful case, the Grover procedure requires $1/\sqrt{P\big(x_m = x^{\star}\big|d_H(x_0|x^{\star}) = \kappa n\big)}$ iterations. Consequently, an upper bound on the total number of steps is 
\begin{equation}
\begin{split}
N_{\mathrm{GW}} \leq & \frac{1}{P\big(d_H(x_0,x^{\star}) =\kappa n\big)\sqrt{P\big(x_m = x^{\star}\big|d_H(x_0|x^{\star}) = \kappa n\big)}}\\
& [\textrm{By Proposition \eqref{gnsfgnsgm}}]\\
\leq &\frac{1}{P(d_0 = \kappa n)\sqrt{P(d_m\leq 0|d_0 = \kappa n)}}\\
& [\textrm{By \eqref{gfnfsgmgsfm} and \eqref{dfgndghmdhg}
}]\\
\leq & \frac{1}{P(E_1)\sqrt{P(E_2)}},\quad \mathrm{if}\quad (2\nu -1)m\geq \kappa n.
\end{split}
\end{equation}

\end{document}